\documentclass[onecolumn,journal]{IEEEtran}
\IEEEoverridecommandlockouts
\usepackage{tikz}
\usepackage[section]{placeins}
\usepackage{bm}
\usepackage{subcaption}
\usepackage{float}
\usepackage{amsmath}
\usepackage{graphicx}
\usepackage{epstopdf}
\usepackage{amssymb}
\usepackage{amssymb,amsmath}
\usepackage{amsgen,amsfonts,amsbsy,amsthm}
\usepackage{latexsym}
\usepackage{times}
\usepackage{algorithm}
\usepackage[noend]{algpseudocode}
\makeatletter
\def\BState{\State\hskip-\ALG@thistlm}
\makeatother
\DeclareMathOperator*{\argmin}{arg\,min}
\DeclareMathOperator*{\argmax}{arg\,max}
\newcommand{\spn}[1]{\mbox{\texttt{span}}\left(#1\right)}
\newcommand{\inprod}[2]{\left\langle #1,#2 \right\rangle}
\newcommand{\bvec}[1]{\bm{#1}}
\newcommand{\real}{\mathbb{R}}

\newcommand{\support}{\mathcal{H}}
\newcommand{\proj}[1]{\mathbf{P}_{#1}}
\newcommand{\dualproj}[1]{\mathbf{P}_{#1}^\perp}
\newcommand{\norm}[1]{\|#1\|_2}
\newcommand{\abs}[1]{\left|#1\right|}
\newtheorem{lem}{Lemma}[section]
\newtheorem{thm}{Theorem}[section]

\title{A Modified  Multiple OLS (m$^2$OLS) Algorithm for Signal Recovery in Compressive Sensing}
\begin{document}
\vspace{-5mm}
\author{Samrat Mukhopadhyay$^1$, {\sl Student Member, IEEE}, Siddhartha Satpathi  $^2$, and Mrityunjoy
Chakraborty$^3$, {\sl Senior Member, IEEE}



\thanks{Authors $^1$ and $^3$ are with the department of Electronics and Electrical Communication
Engineering, Indian Institute of Technology, Kharagpur, INDIA, and author $^2$ is with the department of Electrical Engg., University of Illinois at Urbana Champaign, USA (email : $^1$samratphysics@gmail.com, $^2$sidd.piku@gmail.com, $^3$mrityun@ece.iitkgp.ernet.in).
}}
\IEEEoverridecommandlockouts
\maketitle
\begin{abstract}
Orthogonal least square (OLS) is an important sparse signal recovery algorithm in compressive sensing, which enjoys superior probability of success over other well known recovery algorithms under conditions of correlated measurement matrices. Multiple OLS (mOLS) is a recently proposed improved version of OLS which selects multiple candidates per iteration by generalizing the greedy selection principle used in OLS and enjoys faster convergence than OLS. In this paper, we present a refined version of the mOLS algorithm where at each step of iteration, we first preselect a submatrix of the measurement matrix suitably and then apply the mOLS computations to the chosen submatrix. Since mOLS now works only on a submatrix and  not on the overall matrix, computations reduce drastically. Convergence of the algorithm, however, requires to ensure passage of true candidates through the two stages of preselection and mOLS based selection successively. This paper presents convergence conditions for both noisy and noise free signal models. The proposed algorithm enjoys faster convergence properties similar to mOLS, at a much reduced computational complexity.
\end{abstract}

%
%
%
    %

\begin{IEEEkeywords}
Compressive Sensing, mOLS, restricted isometry property
\end{IEEEkeywords}
\section{Introduction}
\label{sec:intro} Signal recovery in compressive sensing (CS)
requires evaluation of the sparsest solution to an underdetermined
set of equations $\bvec{y}=\bvec{\Phi x}$, where $\bvec{\Phi}\in
\real^{m\times n}\;(m<<n)$ is the so-called measurement matrix and
$\bvec{y}$ is the $m\times 1$ observation vector. It is usually
presumed that the sparsest solution is $K$-sparse, i.e., not more
than $K$ elements of $\bvec{x}$ are non-zero, and also that the
sparsest solution is unique which can be ensured by maintaining
every $2K$ columns of $\bvec{\Phi}$ as linearly independent.
There exist a popular class of algorithms in literature called
greedy algorithms, which obtain the sparsest $\bvec{x}$ by
iteratively constructing the support set of $\bvec{x}$ (i.e., the
set of indices of non-zero elements in $\bvec{x}$) via some greedy
principles. Orthogonal Matching
Pursuit(OMP)~\cite{tropp2007signal} is a prominent algorithm in
this category, which, at each step of iteration, enlarges a
partially constructed support set by appending a column of
$\bvec{\Phi}$ that is most strongly correlated with a residual
vector, and updates the residual vector by projecting $\bvec{y}$
on the column space of the sub-matrix of $\bvec{\Phi}$ indexed by
the updated support set, and then taking the projection error.
Tropp and Gilbert~\cite{tropp2007signal} have shown that OMP can
recover the original sparse vector from a few measurements with
exceedingly high probability when the measurement matrix has i.i.d
Gaussian entries. OMP was extended by Wang \emph{et al}~\cite{wang2012generalized} to the generalized orthogonal matching pursuit (gOMP)where at the indentification stage, multiple columns are selected based on the correlation of the columns of matrix $\bvec{\Phi}$ with the residual vector, which allows gOMP to enjoy faster convergence compared to OMP.

It has, however, been shown recently by  Soussen \emph{et
al}~\cite{soussen2013joint} that the probability of success in OMP
reduces sharply as the correlation between the columns of
$\bvec{\Phi}$ increases, and for measurement matrices with
correlated entries, another greedy algorithm, namely, the
Orthogonal Least Squares (OLS)~\cite{chen1989orthogonal} enjoys
much higher probability of recovery of the sparse signal than OMP.
OLS is computationally similar to OMP except for a more expensive
greedy selection step. Here, at each step of iteration, the
partial support set already evaluated is augmented by an index $i$
which minimizes the energy (i.e., the $l_2$ norm) of the resulting
residual vector.

~An improved version of OLS called multiple OLS (mOLS) has been
proposed recently by  Wang \emph{et al} ~\cite{wang2017recovery},
where unlike OLS, a total of $L$ ($L>1$) indices are appended to
the existing partial support set by suitably generalizing the
greedy principle used in OLS. As $L$ indices are chosen each time,
possibility of selection of multiple ``true'' candidates in each
iteration increases and thus, the probability of convergence in
much fewer iterations than OLS becomes significantly high.

~In this paper, we present a refinement of the mOLS algorithm,
named as modified mOLS (m$^2$OLS), where, at each step of
iteration, we first \textit{pre-select} a total of, say, $N$
columns of $\bvec{\Phi}$ by evaluating the correlation between the
columns of $\bvec{\Phi}$ with the current residual vector and
choosing the $N$ largest (in magnitude) of them. The steps of mOLS
are then applied to this pre-selected set of columns. Here the preselection strategy is identical to the identification strategy of gOMP so that chances of selection of
multiple ``true'' candidates in the pre-selected set is expected to be high. Furthermore, as the mOLS
subsequently works on this preselected set of columns and not on the entire matrix $\bvec{\Phi}$, to determine a subset of $L$ columns ($L<N$), computational costs reduce drastically compared to conventional mOLS. This is
also confirmed by our simulation studies. Derivation of conditions
of convergence for the proposed algorithm is, however, tricky, as
it requires to ensure simultaneous passage of at least one true
candidate from $\bvec{\Phi}$ to the pre-selected set and then,
from the pre-selected set to the mOLS determined subset at every
iteration step. This paper presents convergence conditions of the
proposed algorithm for the cases of both noise free and noisy
observations. It also presents the computational steps of an
efficient implementation of both mOLS and m$^2$OLS, and brings out
the computational superiority of m$^2$OLS over mOLS analytically.
Detailed simulation results in support of the claims made are also
presented.
\section{Preliminaries}
\label{sec:preliminaries}
The following notations have been used throughout the paper :`$t$' in
superscript indicates transposition of matrices / vectors.%
$\bvec{\Phi}\in \mathbb{R}^{m\times n}$ denotes the measurement
matrix ($m<n$) and the $i$ th column of $\bvec{\Phi}$ is denoted
by $\bm{\phi}_i,\ i=1,2,\cdots,\ n$. All the columns of
$\bvec{\Phi}$ are assumed to have unit $l_2$ norm, i.e.,
$\|\bm{\phi}_i\|_2=1$, which is a common assumption in the
literature~\cite{tropp2007signal},~\cite{wang2017recovery}.
$\support$ denotes the set of all the indices $\{1,2,\cdots,\
n\}$. $K$ indicates the sparsity level of $\bvec{x}$, i.e., not
more than $K$ elements of $\bvec{x}$ are non-zero. $T$ denotes the
true support set of $\bvec{x}$, i.e., $T=\{i\in
\support|[\bvec{x}]_i\ne 0\}.$ For any $S\subseteq \support$,
$\bvec{x}_S$ denotes the vector $\bvec{x}$ restricted to $S$,
i.e., $\bvec{x}_S$ consists of those entries of $\bvec{x}$ that
have indices belonging to $S$. Similarly, $\bvec{\Phi}_S$ denotes
the submatrix of $\bvec{\Phi}$ formed with the columns of
$\bvec{\Phi}$ restricted to the index set $S$. If $\bvec{\Phi}_S$
has full column rank of $|S|$ ($|S|<m$), then the Moore-Penrose
pseudo-inverse of $\bvec{\Phi}_{S}$ is given by
$\bvec{\Phi}_{S}^\dagger=(\bvec{\Phi}_{S}^t\bvec{\Phi}_{S})^{-1}\bvec{\Phi}_{S}^t$.
$\proj{S}=\bvec{\Phi}_{S}\bvec{\Phi}_{S}^\dagger$ denotes the
orthogonal projection operator associated with
$span(\bvec{\Phi}_{S})$ and $\dualproj{S}=\bvec{I}-\proj{S}$
denotes the orthogonal projection operator on the orthogonal
complement of $span(\bvec{\Phi}_{S})$. For any set $S\subseteq
\mathcal{H}$, the matrix $\dualproj{S}\bvec{\Phi}$ is denoted by
$\bvec{A}_S$. For a given sparsity order $K$ and a given matrix
$\bvec{\Phi}$, it can be shown that there exists a real, positive
constant $\delta_K$ such that $\bvec{\Phi}$ satisfies the
following ``Restricted Isometry Property (RIP)'' for all
$K$-sparse $\bvec{x}$ :
\begin{align*}
(1-\delta_K)\norm{\bvec{x}}^2\le \norm{\bvec{\Phi x}}^2\le (1+\delta_K)\norm{\bvec{x}}^2.
\end{align*}
The constant $\delta_K$ is called the restricted isometry constant
(RIC)\cite{candes2006robust} of the matrix $\bvec{\Phi}$ for order $K$. Clearly, it is
the minimum such constant for which the RIP is satisfied. Note
that if $\delta_K<1$, $\bvec{x}\ne \bvec{0}$ for a $K$-sparse
$\bvec{x}$ implies $\norm{\bvec{\Phi x}}\ne 0$ and thus,
$\bvec{\Phi x}\ne \bvec{0}$, meaning every $K$ columns of
$\bvec{\Phi}$ are linearly independent. The RIC gives a measure of
near unitariness of $\bvec{\Phi}$ (smaller the RIC is, closer
$\bvec{\Phi}$ will be to being unitary). Convergence conditions of
recovery algorithms in CS are usually given in terms of upper
bounds on the RIC.
\section{Proposed Algorithm}
\begin{table}[ht!]
\centering
\begin{tabular}{p{10cm}}
\centering
\hrulefill
\begin{description}
\item[\textbf{Input:}]\ measurement vector $\bvec{y}\in \real^m$,
sensing matrix $\bvec{\Phi}\in \real^{m\times n}$; sparsity level
$K$; number of indices preselected $N$; number of indices chosen
in identification step, $L(L\le N,\ L\le K)$, prespecified
residual threshold $\epsilon$; \item[\textbf{Initialize:}]$\quad$
counter $k=0$, residue $\bvec{r}^0=\bvec{y}$, estimated support
set, $T^0=\emptyset$, set selected by preselection step
$S^0=\emptyset$, \item[\textbf{While}]($\norm{\bvec{r}^k}\ge
\epsilon\ \mbox{and}\ \ k<K$)
\item[]\  $k=k+1$ \item[]\  {\emph{Preselect:}} $\displaystyle
S^k$ is the set containing indices corresponding to the $N$
largest absolute entries of $\bvec{\Phi}^t \bvec{r}^{k-1}$
\item[]{\emph{Identify:}} $\displaystyle
h^k=\argmin_{\Lambda\subset S^k:|\Lambda|=L}\sum_{i\in
\Lambda}\|\dualproj{T^{k-1}\cup \{i\}}\bvec{y}\|_2^2$ \item[]\
{\emph{Augment:}} $T^k=T^{k-1}\cup h^k$ \item[]\
{\emph{Estimate:}} $\displaystyle
\bvec{x}^k=\argmin_{\bvec{u}:\bvec{u}\in \real^n,\ supp(\bvec{u})=
T^k}\|\bvec{y}-\bvec{\Phi}\bvec{u}\|_2$  \item[]\ \emph{Update:}
$\bvec{r}^k=\bvec{y}-\bvec{\Phi}\bvec{x}^k$\\ \item[]\ (Note :
Computation of $\bvec{x}^k$ for $1\le k \le K$ requires every $LK$
columns of $\bvec{\Phi}$ to be linearly independent which is
guaranteed by the proposed RIC bound)
\item[\textbf{End While}]
\end{description}
\hrulefill
\begin{description}
\item[\textbf{Output:}]$\quad$  estimated support set
$\displaystyle
\hat{T}=\argmax_{\Lambda:|\Lambda|=K}\|\bvec{x}^k_{\Lambda}\|_2$
and $K$-sparse signal $\hat{\bvec{x}}$ satisfying
$\hat{\bvec{x}}_{\hat{T}}=\bvec{\Phi}_{\hat{T}}^\dagger \bvec{y},\
\hat{\bvec{x}}_{\support\setminus\hat{T}}=\mathbf{0}$
\end{description}
\hrulefill
\caption{Proposed m$^2$OLS \textsc{Algorithm}}
\label{tab:m$^2$OLS}
\end{tabular}
\end{table}
The proposed m$^2$OLS algorithm is described in
Table.~\ref{tab:m$^2$OLS}. At any $k$-th step of iteration ($k\ge
1$), assume a residual signal vector $\bvec{r}^{k-1}$ and a
partially constructed support set $T^{k-1}$ have already been
computed ($\bvec{r}^{0}=\bvec{y}$ and $T^0=\emptyset$). In the
\textit{preselection} stage, $N$ columns of $\bvec{\Phi}$ are
identified that have largest (in magnitude) correlations with
$\bvec{r}^{k-1}$ by picking up the $N$ largest absolute entries of
$\bvec{\Phi}^t \bvec{r}^{k-1}$, and the set $S^k$ containing the
corresponding indices is selected. This is followed by the
\textit{identification} stage, where $\sum_{i\in
\Lambda}\|\dualproj{T^{k-1}\cup \{i\}}\bvec{y}\|_2^2$ is evaluated
for all subsets $\Lambda$ of $S^k$ having $L$ elements, and
selecting the subset $h^k$ for which this is minimum. This is the
greedy selection stage, which is carried out in practice
~\cite{wang2017recovery} by computing $\frac{|\bm{\phi}_i^t
\bvec{r}^{k-1}|}{\|\dualproj{T^{k-1}}\bm{\phi}_i\|_2}$ for all
$i\in S^k$ and selecting the indices corresponding to the $L$
largest of them. The partial support set is then updated to $T^k$
by taking set union of $T^{k-1}$ and $h^k$, and the residual
vector is updated to $\bvec{r}^k$ by computing
$\dualproj{T^{k}}\bvec{y}$.

Note that in conventional mOLS algorithm, at a $k$-th step of
iteration ($k\ge 1$), one has to compute
$\frac{|\bm{\phi}_i^t\bvec{r}^{k-1}|}{\|\dualproj{T^{k-1}}\bm{\phi}_i\|_2}$
for all $i \in \support\setminus T^{k-1}$, involving a total of
$n-(k-1)L$ columns, i.e., $\bm{\phi}_i$'s. In contrast, in the
proposed m$^2$OLS algorithm, the above computation is restricted
only to the preselected set of $N$ elements, which results in
significant reduction of computational complexity.
\subsection{Lemmas (Existing)}
The following lemmas will be useful for the analysis of the proposed algorithm.
\begin{lem}[Monotonicity, Lemma 1 of~\cite{dai2009subspace}]
\label{lem:monotonicity}
If a measurement matrix satisfies RIP of orders $K_1,K_2$ and $K_1\le K_2$, then $\delta_{K_1}\le \delta_{K_2}$.
\end{lem}
\begin{lem}[Consequence of RIP~\cite{needell2009cosamp}]
\label{lem:max-min-eigenvalues-phi_t*phi}
For any subset $\Lambda\subseteq \mathcal{H}$, and for any vector $\bvec{u}\in \real^n$,\begin{align*}
(1-\delta_{|\Lambda|})\norm{\bvec{u}_\Lambda}\le \norm{\bvec{\Phi}^t_\Lambda\bvec{\Phi}_\Lambda\bvec{u}_\Lambda}\le (1+\delta_{|\Lambda|})\norm{\bvec{u}_\Lambda}.
\end{align*}
\end{lem}
\begin{lem}[Proposition 3.1 in~\cite{needell2009cosamp}]
\label{lem:upper-bound-phi^t-u}
For any $\Lambda\subseteq \support$, and for any vector $\bvec{u}\in \real^m$\begin{align*}
\norm{\bvec{\Phi}_\Lambda^t\bvec{u}}\le \sqrt{1+\delta_{|\Lambda|}}\norm{\bvec{u}}.
\end{align*}
\end{lem}
%

%
%
\begin{lem}[Lemma 1 of \cite{dai2009subspace}]
\label{lem:orthogonal-sparse1} If $\bvec{x}\in \real^{n}$ is a
vector with support $S_1$, and $S_1\cap S_2=\emptyset$, then,
$$\norm{\bvec{\Phi}_{S_2}^t\bvec{\Phi}\bvec{x}}\le
\delta_{|S_1|+|S_2|} \norm{\bvec{x}}.$$
\end{lem}
\begin{lem}[Lemma 3 of ~\cite{satpathi2013improving}]
\label{lem:orthogonal-sparse2} If $I_1,I_2\subset \support$ such
that $I_1\cap I_2=\emptyset$ and $\delta_{|I_2|}<1$, then,
$\forall \bvec{u}\in \real^n$ such that $supp(\bvec{u})\subseteq
I_2$,
\begin{align*}
\left(1-\left(\frac{\delta_{|I_1|+|I_2|}}{1-\delta_{|I_1|+|I_2|}}\right)^2\right)\norm{\bvec{\Phi}\bvec{u}}^2\le
\norm{\bvec{A}_{I_1}\bvec{u}}^2\le
(1+\delta_{|I_1|+|I_2|})\norm{\bvec{\Phi u}}^2,
\end{align*}
and,
\begin{align*}
\left(1-\frac{\delta_{|I_1|+|I_2|}}{1-\delta_{|I_1|+|I_2|}}\right)\norm{\bvec{u}}^2\le \norm{\bvec{A}_{I_1}\bvec{u}}^2\le (1+\delta_{|I_1|+|I_2|})\norm{\bvec{u}}^2.
\end{align*}
\end{lem}
\section{Signal recovery using m$^2$OLS algorithm}
\label{sec:signal-recovery-theoretical-conditions} In this
section, we obtain convergence conditions for the proposed
m$^2$OLS algorithm. In particular, we derive conditions for
selection of at least one correct index at each iteration, which
guarantees recovery of a $K$-sparse signal by the m$^2$OLS
algorithm in a maximum of $K$ iterations.

Unlike mOLS, proving convergence is, however, trickier in the
proposed m$^2$OLS algorithm because of the presence of two
selection stages at every iteration, namely, preselection and
identification. In order that the proposed algorithm converges in
$K$ steps or less, it is essential to ensure that at each step of
iteration, at least one true support index $i$ first gets selected
in $S^k$ and then, gets passed on from $S^k$ to $h^k$. In the
following, we present the convergence conditions for m$^2$OLS in
two cases, with and without the presence of measurement noise. For
the noiseless measurement model the measurement vector $\bvec{y}$
satisfies $\bvec{y=\Phi x}$, with a \emph{unique} $K$-sparse
vector $\bvec{x}$. For the noisy measurement model, the
measurement vector is assumed to be contaminated by an additive
noise vector, i.e., $\bvec{y=\Phi x+e}$. The convergence
conditions for noiseless and noisy cases are given in Theorems
~\ref{thm:noiseless-recovery} and Theorem~\ref{thm:noisy-recovery}
below. Both these theorems use Lemma~\ref{lem:identification},
which in turn uses the following definition : $\tilde{T}^K=\{i\in
H|\bvec{\phi_i}\in span(\bvec{\Phi}_{T^k})\}$. Note that
$T^k\subseteq \tilde{T}^K$ and for $i\in \tilde{T}^K$,
$\|\dualproj{T^k}\bvec{\phi_i}\|_2=0$,
$\inprod{\bvec{\phi_i}}{\bvec{r}^k}=0$. It should be mentioned that the first result of Lemma~\ref{lem:identification} is not any new result, and similar result has already been discussed in the context of OLS~\cite{soussen2013joint},~\cite{rebollo2002optimized}, and mOLS~\cite{wang2017recovery}. However, we provide an additional observation in the following lemma that concerns the identification step of m$^2$OLS.
\begin{lem}
\label{lem:identification} At the $(k+1)th$ iteration, the
identification step chooses the set \begin{align*}
h^{k+1}=&\argmax_{\Lambda:\Lambda\subset
S^{k+1},|\Lambda|=L}\sum_{i\in \Lambda}a_i^2,
\end{align*}
\ \\
where $a_i=
\frac{|\inprod{\bvec{\phi_i}}{\bvec{r}^k}|}{\|\dualproj{T^k}\bvec{\phi_i}\|_2}$
if $i\in S^{k+1}\setminus \tilde{T}^K$, and $a_i=0$ for $i\in
S^{k+1}\cap \tilde{T}^K$. Further, if
\begin{align*}
g^{k+1}=&\argmax_{\Lambda:\Lambda\subset
S^{k+1},|\Lambda|=L}\sum_{i\in \Lambda}a_i,
\end{align*}
\ \\
then, $\sum_{i\in h^{k+1}}a_i=\sum_{i\in g^{k+1}}a_i$.
\begin{proof}
The first part of this lemma is a direct consequence of
Proposition 1 of~\cite{wang2017recovery}. For the second part, let
$l\in h^{k+1}$ be an index, so that, $a_l\le a_r,\, \forall r\in
h^{k+1}$ (i.e. $a_l=\min\{a_r|\ r\in h^{k+1}\}$). Clearly, $a_l\ge
a_j\ \forall j\in S^{k+1}\setminus h^{k+1}$, as otherwise, if
$\exists\ a_j\in S^{k+1}\setminus h^{k+1}$ so that $a_l <a_j$, we
have $a_l^2<a_j^2$. Then constructing the set $H^{k+1}$ as
$H^{k+1}=h^{k+1}\cup \{j\}\setminus \{l\}$, we have, $\sum_{i\in
h^{k+1}}a_i^2<\sum_{i\in H^{k+1}}a_i^2$, which is a contradiction.
~The above means that $\forall i\in h^{k+1}$, $a_i\ge a_j,\
\forall j\in S^{k+1}\setminus h^{k+1}$. Thus, for any $S\subseteq
S^{k+1},\ \abs{S}=L$, $\sum_{i\in h^{k+1}}a_i\ge \sum_{i\in
S}a_i$, and thus, $\sum_{i\in h^{k+1}}a_i\ge \sum_{i\in
g^{k+1}}a_i$. Again, from the definition of $g^{k+1}$, $\sum_{i\in
g^{k+1}}a_i\ge \sum_{i\in h^{k+1}}a_i$. This proves the desired
equality.
\end{proof}
\end{lem}
Our main results regarding the signal recovery performance of the m$^2$OLS algorithm is stated in the following two theorems.
\begin{thm}
\label{thm:noiseless-recovery} The m$^2$OLS algorithm can recover
a $K$ sparse vector $\bvec{x}\in \real^n$ perfectly from the
measurement vector $\bvec{y}=\bvec{\Phi x},\ \bvec{y}\in \real^m,\
m<n$ within $K$ iterations, if
\begin{align}
\delta_{LK+N-L+1}<\frac{\sqrt{L}}{\sqrt{K+L}+\sqrt{L}}
\label{eq:final-condition2}
\end{align} is satisfied by matrix $\bvec{\Phi}$.
\end{thm}
\begin{proof}
Given in Appendix~\ref{sec:appendix-proof-thms}.
\end{proof}
To describe recovery performance of m$^2$OLS in presence of noise,
we use the following performance measures
~\cite{wang2017recovery}:
\begin{itemize}
\item $snr:=\frac{\norm{\bvec{\Phi x}}^2}{\norm{\bvec{e}}^2}$,
\item minimum-to-average-ratio (MAR)~\cite{fletcher2012orthogonal}, $\kappa=\frac{\min_{j\in T}|x_j|}{\norm{\bvec{x}}/\sqrt{K}}$.
\end{itemize}
\begin{thm}
\label{thm:noisy-recovery} Under the noisy measurement model,
m$^2$OLS is guaranteed to collect all the indices of the the true
support set $T$ within $K$ iterations, if the sensing matrix
$\bvec{\Phi}$ satisfies equation~\eqref{eq:final-condition2} and
the $snr$ satisfies the following condition:
\begin{align}
\label{eq:noisy-recovery-condition}
\sqrt{snr}\;>\frac{(1+\delta_{R})(\sqrt{L}+\sqrt{K})\sqrt{K}}{\kappa\left(\sqrt{L(1-2\delta_{R})}-\delta_{R}\sqrt{K}\right)},
\end{align}
where $R = LK + N - L + 1$.
\end{thm}
\begin{proof}
Given in Appendix~\ref{sec:appendix-proof-thms}.
\end{proof}
Note that the m$^2$OLS algorithms reduces to the gOMP algorithm when $N=L$. Theorem~\ref{thm:noiseless-recovery} suggests that for $N=L$, the m$^2$OLS algorithm can recover the true support of any $K$-sparse signal from noiseless measurements within $K$ if the sensing matrix satisfies $\delta_{NK+1}<\frac{1}{\sqrt{K/N+1}+1}$, where $N\ge 1$. Recently Wen~\emph{et al}~\cite{wen2017novel} have established that, with $N\ge 1$, $\delta_{NK+1}<\frac{1}{\sqrt{K/N+1}}$ is a sharp sufficient condition for gOMP to exactly recover $K$-sparse signals from noiseless measurements within $K$ iterations. We see that for large $K/N$ ratio, $\sqrt{K/N+1}+1\approx \sqrt{K/N+1}$, which shows that the bound provided in Theorem~\ref{thm:noiseless-recovery} is nearly sharp when $N=L$. Moreover, in our analysis  When $N>1$, and $L=1$, the bound in Theorem~\ref{thm:noiseless-recovery} reduces to $\delta_{K+N}<\frac{1}{\sqrt{K+1}+1}$, whereas, the recent paper~\cite{wen2017nearly} suggests the sufficient condition $\delta_{K+1}<\frac{1}{\sqrt{K+1}}$ for the OLS to recover perfectly a $K$-sparse signal from noiseless measurements within $K$ iterations. Again, the right hand side of the bound suggested in Theorem~\ref{thm:noiseless-recovery} is very close to the one established in~\cite{wen2017nearly} for large $K$. However, the left hand side of the inequality contains $\delta_{K+N}$ in our case, which can be much larger than $\delta_{K+1}$, and thus makes our condition stricter than that of the one obtained in~\cite{wen2017nearly}. However, since in m$^2$OLS the operations of mOLS are performed on a smaller preselected set of indices to reduce computational cost, intuitively it is expected that the sensing matrix must satisfy some stricter RIP condition in order to yield recovery performance competitive to that of mOLS.

Our proof of the theorems~\ref{thm:noiseless-recovery} and~\ref{thm:noisy-recovery} mainly follows the ideas of the analysis of gOMP~\cite[Theorem 4.2]{wang2012generalized} and mOLS~\cite[Theorem 3]{wang2017recovery} where sufficient conditions for recovery of signal support of a $K$-sparse signal from noisy measurements by running the algorithm no more than $K$ iterations were established. The proof uses mathematical induction, where  we first find out a sufficient condition for success by m$^2$OLS in the first iteration, and then assuming that m$^2$OLS is successful in each of the previous $k(1\le k\le K-1)$ iterations, we find out conditions sufficient for success at the $(k+1)^{\mathrm{th}}$ iteration. However, unlike gOMP or mOLS, finding out the condition for success at any iteration for m$^2$OLS requires ensuring that first at least one true index is captured in the preselection step, which is identical to the gOMP identification step, and then further ensuring that at least one of these captured true indices should be recaptured by the identification step, which is identical to the mOLS identification step. Thus any iteration of m$^2$OLS is successful if the sufficient conditions for both these steps are satisfied. Finally, the sufficient conditions for any general iteration $k(2\le k \le K)$, and the condition for iteration $1$ is combined to obtain the final condition for successful recovery of support within $K$ iterations.

The steps in our proof are partly similar to the steps in the proof of Theorem 3 in Wang \emph{et al}~\cite{wang2017recovery}, and we have frequently used certain steps in the proof of Theorem 1 in the paper of Li \emph{et al}~\cite{li2015sufficient}, specifically \cite[Eq.(25),(26)]{li2015sufficient} and Eq. (9) of Satpathi \emph{et al}~\cite{satpathi2013improving}. However, our analysis have differed from these analysis during the analysis of first iteration of m$^2$OLS, where unlike Wang \emph{et al}~\cite{wang2012generalized}, and Li \emph{et al}~\cite{li2015sufficient} we have given the analysis both for the cases $1\le N\le K$, and $N>K$. Furthermore, in the identification step we have used Lemma~\ref{lem:orthogonal-sparse2} which have produced bound on $\norm{\dualproj{T^k}\bvec{\phi}_i}^2$ tighter than the one that follows from \cite[Eq.(E.7)]{wang2017recovery}, which has been used in the subsequent analysis of m$^2$OLS.  

\section{Comparative analysis of computational complexities of mOLS and m$^2$OLS}
\label{sec:complexity-analysis}
%
By restricting the steps of mOLS to a pre-selected subset of
columns of $\bvec{\Phi}$, the proposed m$^2$OLS algorithm achieves
considerable computational simplicity over mOLS. In this section,
we analyze the computational steps involved in both mOLS and
m$^2$OLS at the $(k+1)^\mathrm{th}$ iteration (i.e., assuming that
$k$ iterations of either algorithm have been completed), and
compare their computational costs in terms of number of floating
point operations (flops) required.
\subsection{Computational steps  of mOLS (in iteration $k+1$)}
\label{sec:computational-cost-mols} \noindent\textbf{Step $1$}
(\textbf{Absolute correlation calculation}) : Here
$\abs{\inprod{\bvec{\phi}_i}{\bvec{r}^k}}$ is calculated $\forall
i\in \support\setminus T^k$, where the vector $\bvec{r}^k$ was
precomputed at the end of the $k^\mathrm{th}$ step. We initialize
$\bvec{r}^0 = \bvec{y}$.\\
\textbf{Step $2$} (\textbf{Identification}) : In this step, mOLS
first calculates the ratios
$\frac{\abs{\inprod{\bvec{\phi}_i}{\bvec{r}^k}}}
{\norm{\dualproj{T^k}\bvec{\phi}_i}},\ \forall i\in
\support\setminus T^k$. Since $\forall i\in \support \setminus
T^k$, the numerator was calculated in Step 1, only the denominator
needs to be calculated. However, as will be discussed later, at
the end of each $k^\mathrm{th}$ step, the norms
$\norm{\dualproj{T^{k}}\bvec{\phi}_i},\ i\in \support\setminus
T^{k}$ are calculated and stored, which provides the denominators
in the above ratios. This means, the above computation requires
simply a division operation per ratio and a total of $(n-Lk)$
divisions. This step is followed by finding the $L$ largest of the
above ratios, and appending the corresponding columns to the
previously estimated subset of columns, $\bvec{\Phi}_{T^k}$,
thereby generating $\bvec{\Phi}_{T^{k+1}}$ (for $k=0$,
$T^k=\emptyset$ and thus, $\bvec{\Phi}_{T^k}=\emptyset$).
\\
\textbf{Step $3$} (\textbf{Modified Gram Schmidt}) :  This step
finds an orthonormal basis for $\spn{\bvec{\Phi}_{T^{k+1}}}$.
Assuming that an orthonormal basis $\{\bvec{u}_1,\ \cdots,\
\bvec{u}_{\abs{T^k}}\}$ for $\spn{\bvec{\Phi}_{T^{k}}}$ has
already been computed at the $k^\mathrm{th}$ step, an efficient
way to realize this will be to employ the well known Modified Gram
Schmidt (MGS) procedure \cite{golub2012matrix}, which first
computes $\dualproj{T^k}\bvec{\phi}_i,\ i\in h^{k+1}$ using the
above precomputed orthonormal basis and then, orthonormalizes them
recursively, generating the orthonormal set $\{\bvec{u}_{\abs{T}^k
+ 1},\cdots,\ \bvec{u}_{\abs{T^{k+1}}}\}$.
\\
\textbf{Step $4$} (\textbf{Precomputation of orthogonal projection
error norm}) : At the $(k+1)^\mathrm{th}$ step, after MGS is used
to construct an orthonormal basis for
$\spn{\bvec{\Phi}_{T^{k+1}}}$, the norms
$\norm{\dualproj{T^{k+1}}\bvec{\phi}_i},\ i\in \support \setminus
T^{k+1}$, are computed using the following recursive relation, for
use in the identification step of $(k+2)^\mathrm{th}$
step:\begin{align} \label{eq:phi-proj-error-calculation}
\norm{\dualproj{T^{k+1}}\bvec{\phi}_i}^2 & =
\norm{\dualproj{T^{k}}\bvec{\phi}_i}^2 - \sum_{j=\abs{T^{k}} +
1}^{\abs{T^{k+1}}}\abs{\inprod{\bvec{\phi}_i}{\bvec{u_j}}}^2.
\end{align}
\\
\textbf{Step $4$} (\textbf{Calculation of $\bvec{r}^{k+1}$}) :
Finally mOLS calculates the residual vector $\bvec{r}^{k+1}$ as
follows:
\begin{align} \label{eq:calculation-of-residual} \bvec{r}^{k+1} =
\bvec{r}^k - \sum_{j=\abs{T^k} +
1}^{\abs{T^{k+1}}}\inprod{\bvec{y}}{\bvec{u}_j}\bvec{u}_j.
\end{align}
%
\subsection{Computational steps of m$^2$OLS (in iteration $k+1$)}
\label{sec:computational-cost-m2ols} \noindent\textbf{Step $1$}
(\textbf{Preselection}): In this step, similar to mOLS, the
absolute correlations $\abs{\inprod{\bvec{\phi}_i}{\bvec{r}^k}}$
are calculated using the vector $\bvec{r}^k$ which is precomputed
at the end of the $k^\mathrm{th}$ step.
Then the indices corresponding to the $N$
largest absolute correlations are selected to form the set
$S^{k+1}$.
\\
\textbf{Step $2$} (\textbf{Identification}): The identification
step calculates the ratios
$\frac{\abs{\inprod{\bvec{\phi}_i}{\bvec{r}^k}}}{\norm{\dualproj{T^k}\bvec{\phi}_i}},\
\forall i\in S^{k+1}$, for which the numerators are already known
from Step $1$ and the denominators are calculated as per the
following:
\begin{equation}
\norm{\dualproj{T^k}\bvec{\phi}_i}^2~=~ \norm{\bvec{\phi}_i}^2 -
\sum_{j=1}^{\abs{T^k}}\abs{\inprod{\bvec{\phi}_i}{\bvec{u_j}}}^2,\label{eq:m2ols-denom-norm-computation}
\end{equation}
where, as in mOLS, $\{\bvec{u}_{1},\cdots,\
\bvec{u}_{\abs{T}^k}\}$ is the orthonormal basis formed for
$\spn{\bvec{\Phi}_{T^k}}$ using MGS at step $k$. For the first
step $k=0$, $T^k=\emptyset$, and thus
$\norm{\dualproj{T^0}\bvec{\phi}_i} = \norm{\bvec{\phi}_i},\ i\in
\support.$ It is assumed that the norms $\norm{\bvec{\phi}_i}$ are
all precomputed (taken to be unity in this paper).
This computation is followed by $N$ divisions as required to form
the above ratios. Following this, the indices corresponding to the
largest $L$ of the $N$ ratios are determined and the corresponding
columns are appended to the previously estimated set of columns
$\bvec{\Phi}_{T^k}$ to obtain $\bvec{\Phi}_{T^{k+1}}.$
\\
\textbf{Step $3$} (\textbf{Modified Gram Schmidt}): This step is identical to the MGS step in mOLS, which
generates an orthonormal basis for $\spn{\bvec{\Phi}_{T^{k+1}}}$.
\\
\textbf{Step $4$} (\textbf{Computation of $\bvec{r}^{k+1}$}): As
in mOLS, the residual $\bvec{r}^{k+1}$ is updated using
Eq.~\eqref{eq:calculation-of-residual}.
%
\\\\
\textbf{Comparison between computational complexities of mOLS and
m$^2$OLS: } While certain operations like MGS, computation of
absolute correlation and the residual $\bvec{r}^{k}$ are same in
both mOLS and m$^2$OLS, the major computational difference between
them lies in the following : at the end of every
$(k+1)^\mathrm{th}$ step, the mOLS computes
$\norm{\dualproj{T^{k+1}}\bvec{\phi}_i}^2\;\forall i\in \support
\setminus T^{k+1}$ using recursion of the form
\eqref{eq:phi-proj-error-calculation}. If computation of
$\abs{\inprod{\bvec{\phi}_i}{\bvec{u_j}}}^2$ has a complexity of
$r$ flops, this requires a total of $(n-L(k+1)
 )(Lr+1)$ flops.
Additionally, mOLS requires $(n-Lk)$ divisions to compute the
ratios $\frac{\abs{\inprod{\bvec{\phi}_i}{\bvec{r}^k}}}
{\norm{\dualproj{T^k}\bvec{\phi}_i}},\ \forall i\in
\support\setminus T^k$. The m$^2$OLS, on the other hand,
calculates $\abs{\inprod{\bvec{\phi}_i}{\bvec{u_j}}}^2$ only for
$i\in S^{k+1}$, following \eqref{eq:m2ols-denom-norm-computation},
involving at the most just $N$ and not $(n-Lk)$ columns. The
summation on the RHS of \eqref{eq:m2ols-denom-norm-computation},
however, has $Lk$ terms, meaning this step requires a total of
$N(Lkr+2)$ flops (including the $N$ divisions to compute
$\frac{\abs{\inprod{\bvec{\phi}_i}{\bvec{r}^k}}}
{\norm{\dualproj{T^k}\bvec{\phi}_i}},\ \forall i\in S^{k+1}$).
Clearly, mOLS will require more computations than m$^2$OLS as long
as $1<\frac{(n-Lk)(Lr+2)-L(Lr+1)}{N(Lkr+2)}\approx
\frac{n-Lk-L}{Nk}$, or, equivalently, for $k<\frac{n-L}{N+L}$.
Thus, for large $n$ and / or small $K$, as $k\le K$, the mOLS will
have significantly higher computational overhead as compared to
m$^2$OLS at each iteration $k$ and the difference in cumulative
computational cost over all iterations put together will be huge.
Even when the sparsity $K$ is larger ($2K<m$), the actual number
of iterations, say, $J$ required for convergence by both mOLS and
m$^2$OLS is usually much less than $K$ and the above difference
continues. In case of large $K$ and $J$ close to $K$, the mOLS
will require more computations than m$^2$OLS for $k$ upto certain
value, beyond which m$^2$OLS will start having more computations
and thus, the difference in cumulative computational cost between
mOLS and m$^2$OLS will start reducing with $k$. This means, for
large $K$, we have a reasonably large range of $J$ for which the
overall computational cost of mOLS remains substantially higher
than that of m$^2$OLS. The above comparative assessment of mOLS
and m$^2$OLS in terms of computations required is also validated
by simulation studies as presented in the next section.
\section{Simulation results}
\label{sec:simulation-results}
%
%
\begin{figure}[t!]
\begin{subfigure}{.5\textwidth}
\centering
\includegraphics[height=1.5in,width=2.5in]{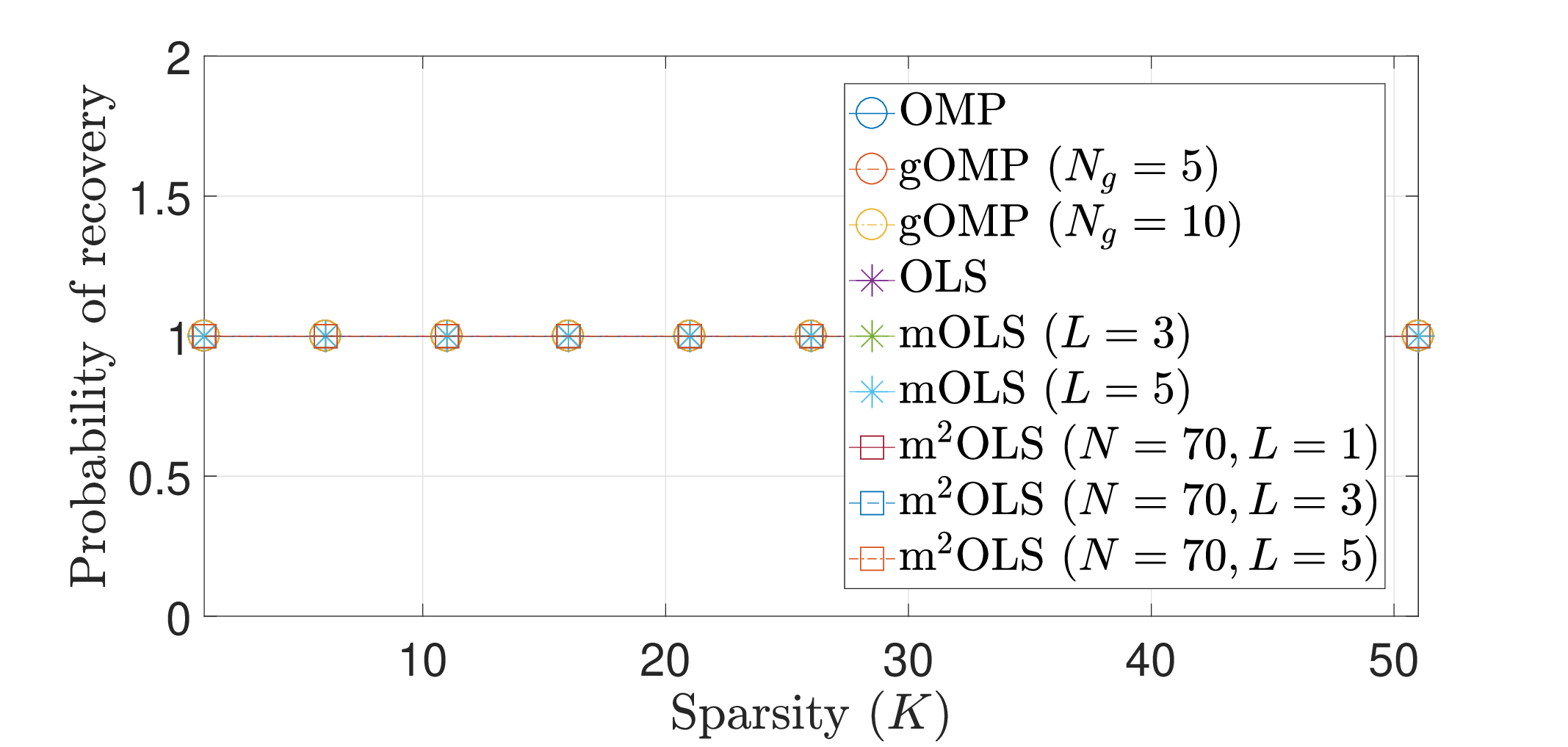}
\caption{$\tau=0$}
\label{fig:prob_K_T=0}
\end{subfigure}
\begin{subfigure}{.5\textwidth}
\centering
\includegraphics[height=1.5in,width=2.5in]{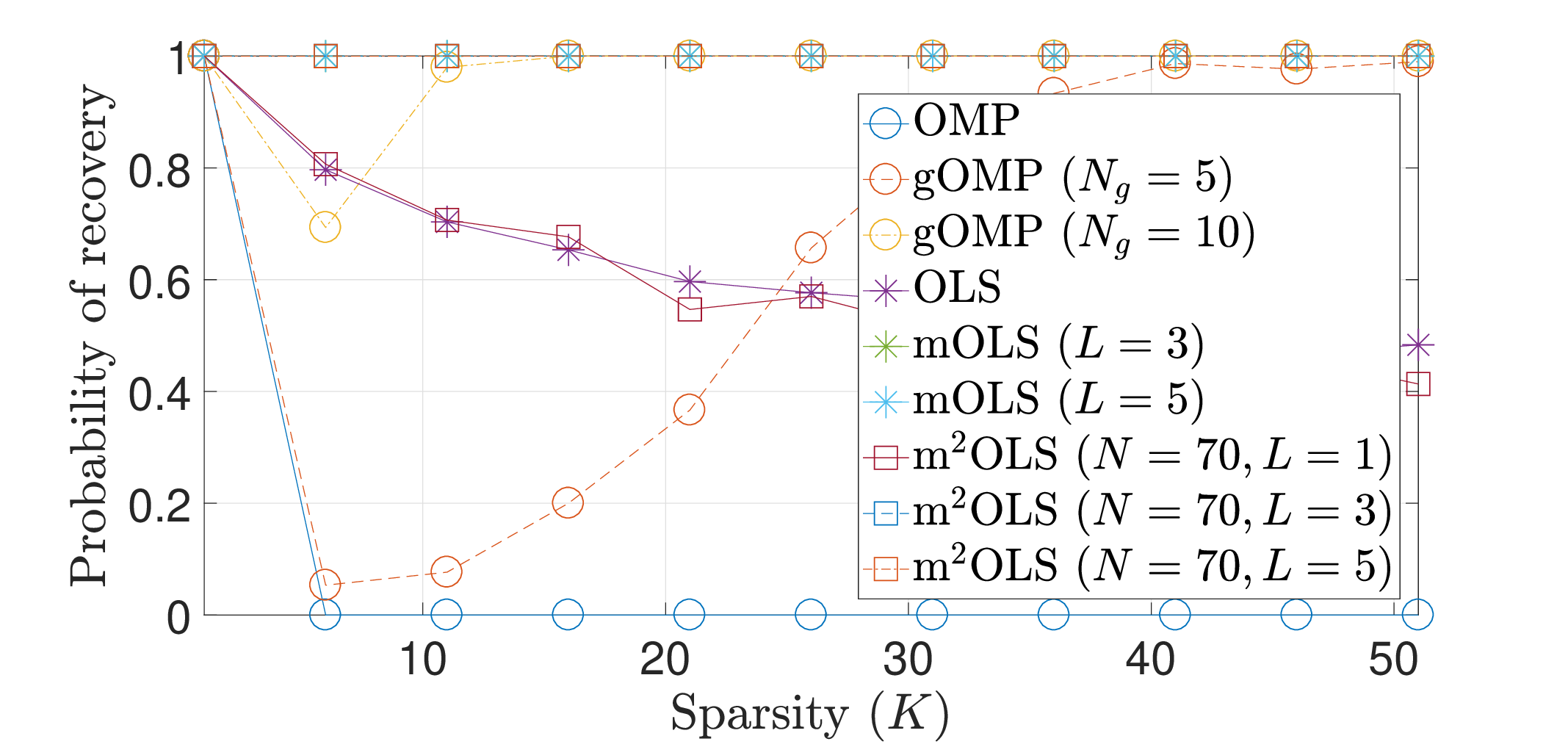}
\caption{$\tau=8$}
\label{fig:prob_K_T=8}
\end{subfigure}
\caption{Recovery probability vs sparsity.}
\label{fig:recovery-prob_vs_K}
\end{figure}
\begin{figure}[t!]
\begin{subfigure}{.5\textwidth}
\centering
\includegraphics[height=1.5in,width=2.5in]{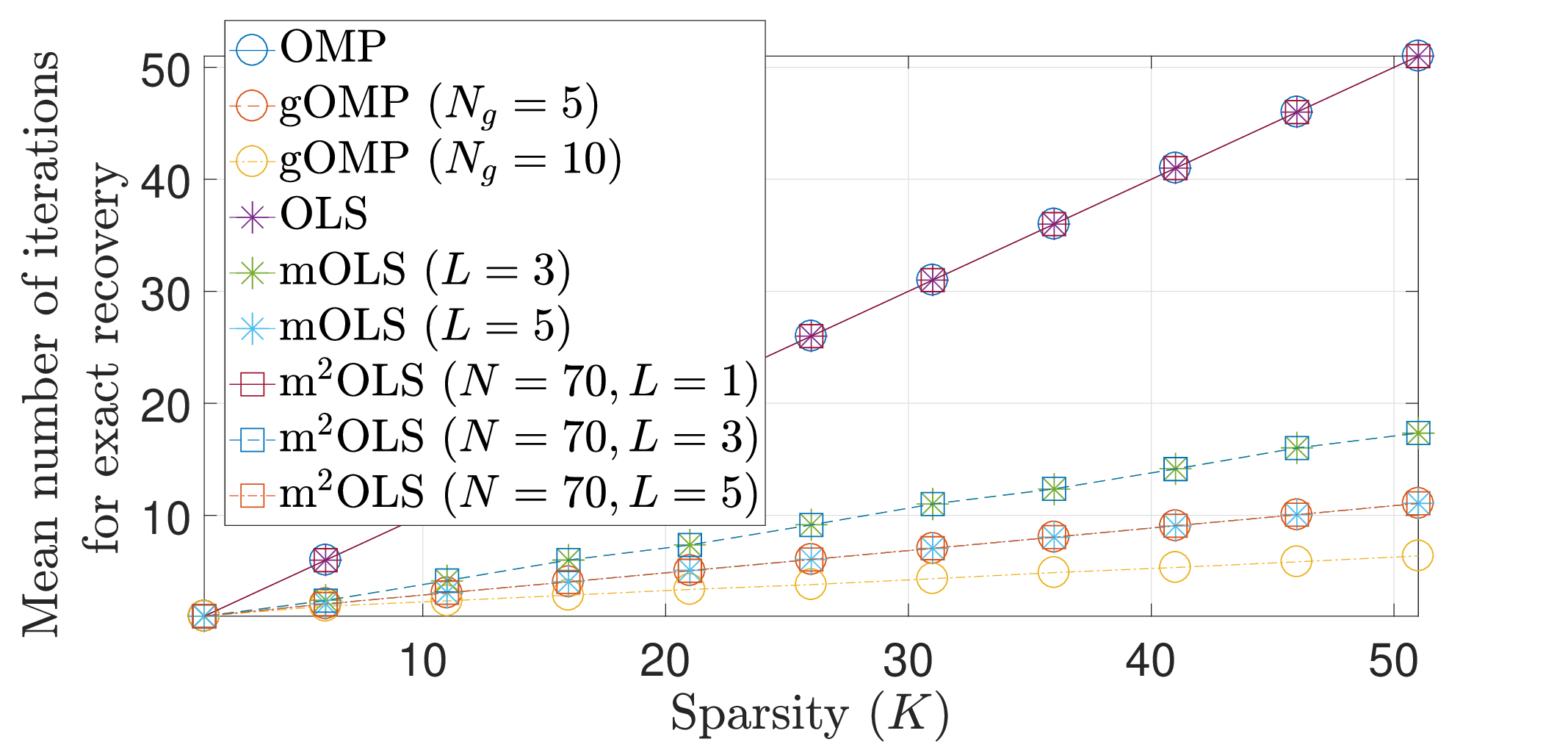}
\caption{$\tau=0$}
\label{fig:no_iteration_T=0}
\end{subfigure}
\begin{subfigure}{.5\textwidth}
\centering
\includegraphics[height=1.5in,width=2.5in]{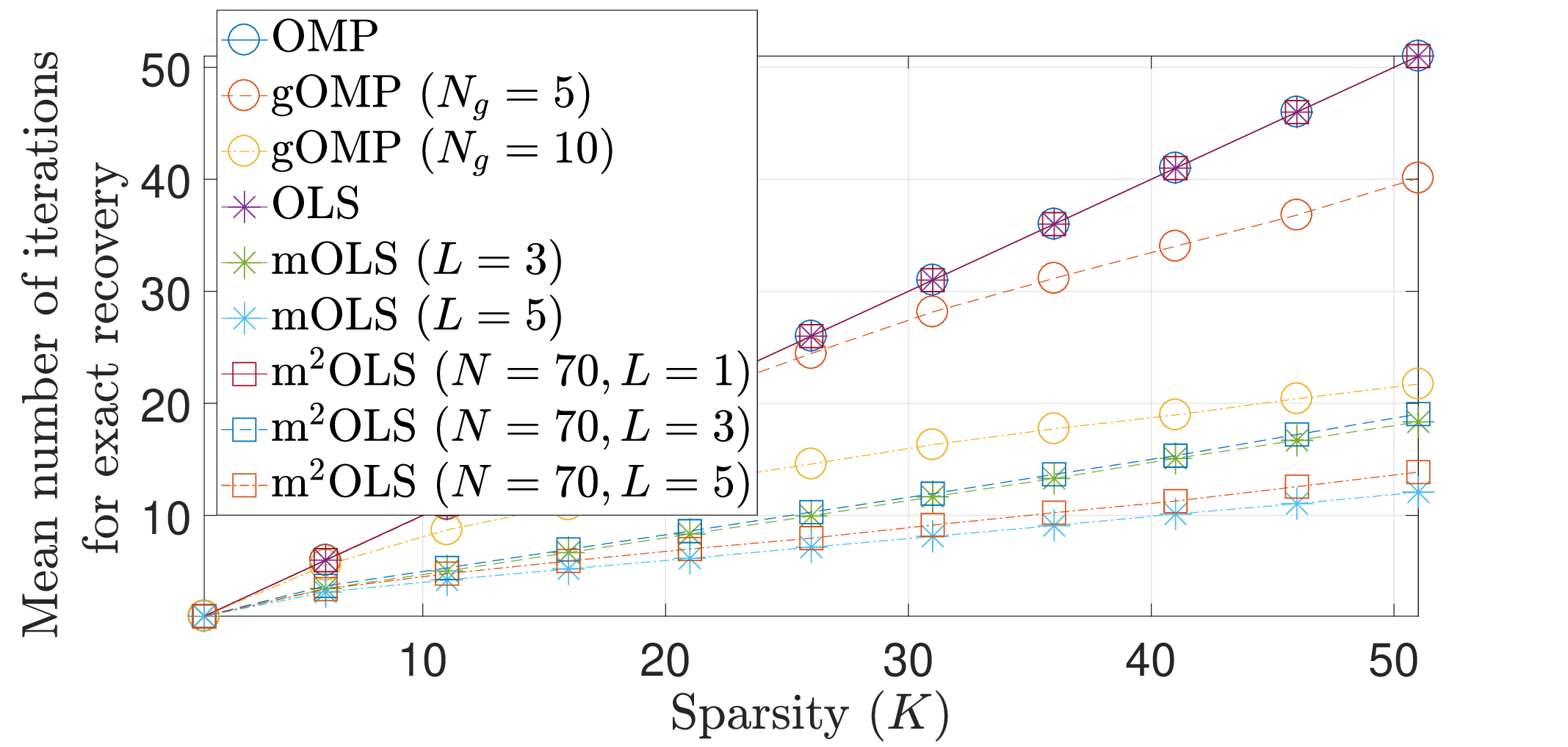}
\caption{$\tau=8$}
\label{fig:no_iteration_T=8}
\end{subfigure}
\caption{No. of iterations for exact recovery vs sparsity.}
\label{fig:no_iteration_sparsity}
\end{figure}
\begin{figure}[t!]
\begin{subfigure}{.5\textwidth}
\centering
\includegraphics[height=1.5in,width=2.5in]{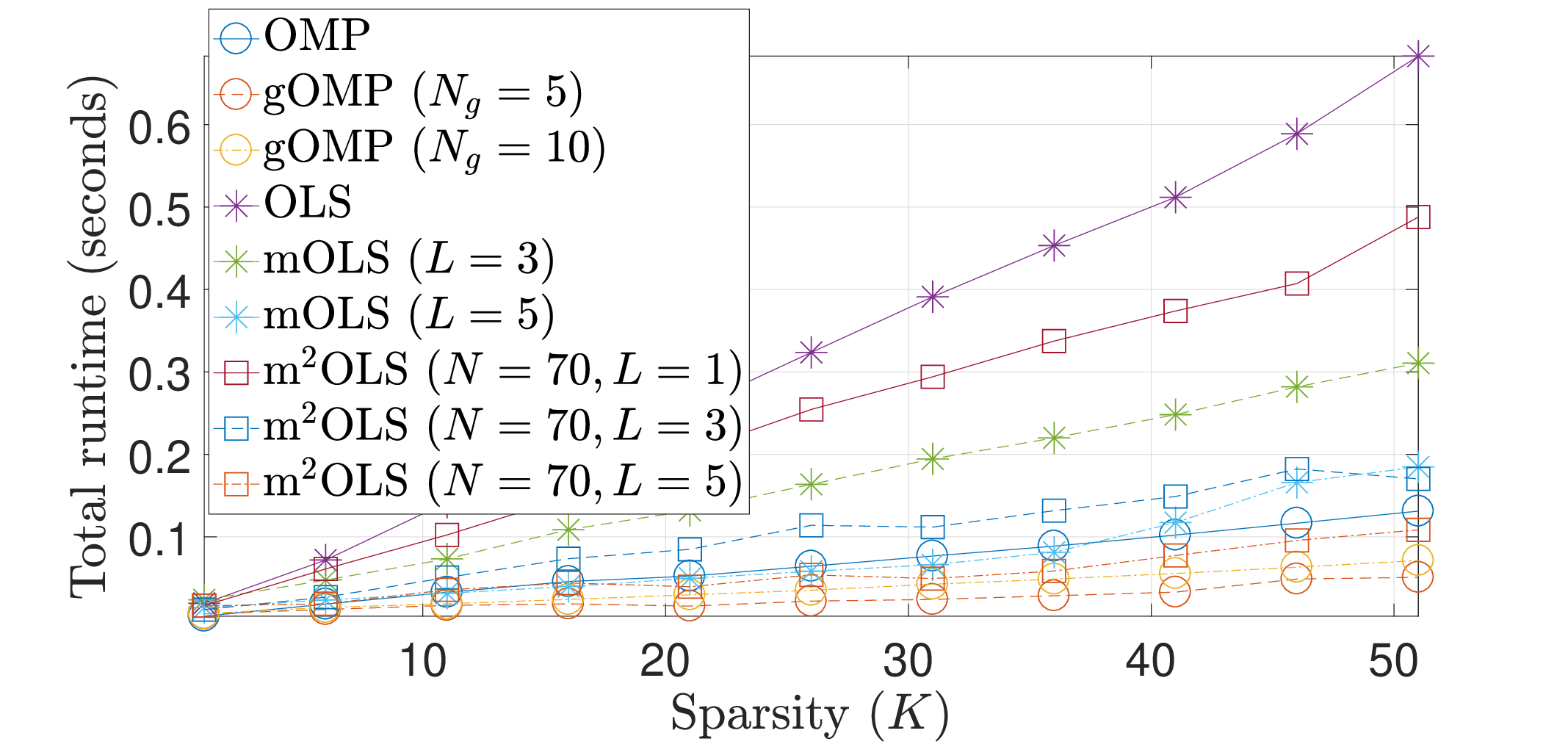}
\caption{$\tau=0$}
\label{fig:runtime_T=0}
\end{subfigure}
\begin{subfigure}{.5\textwidth}
\centering
\includegraphics[height=1.5in,width=2.5in]{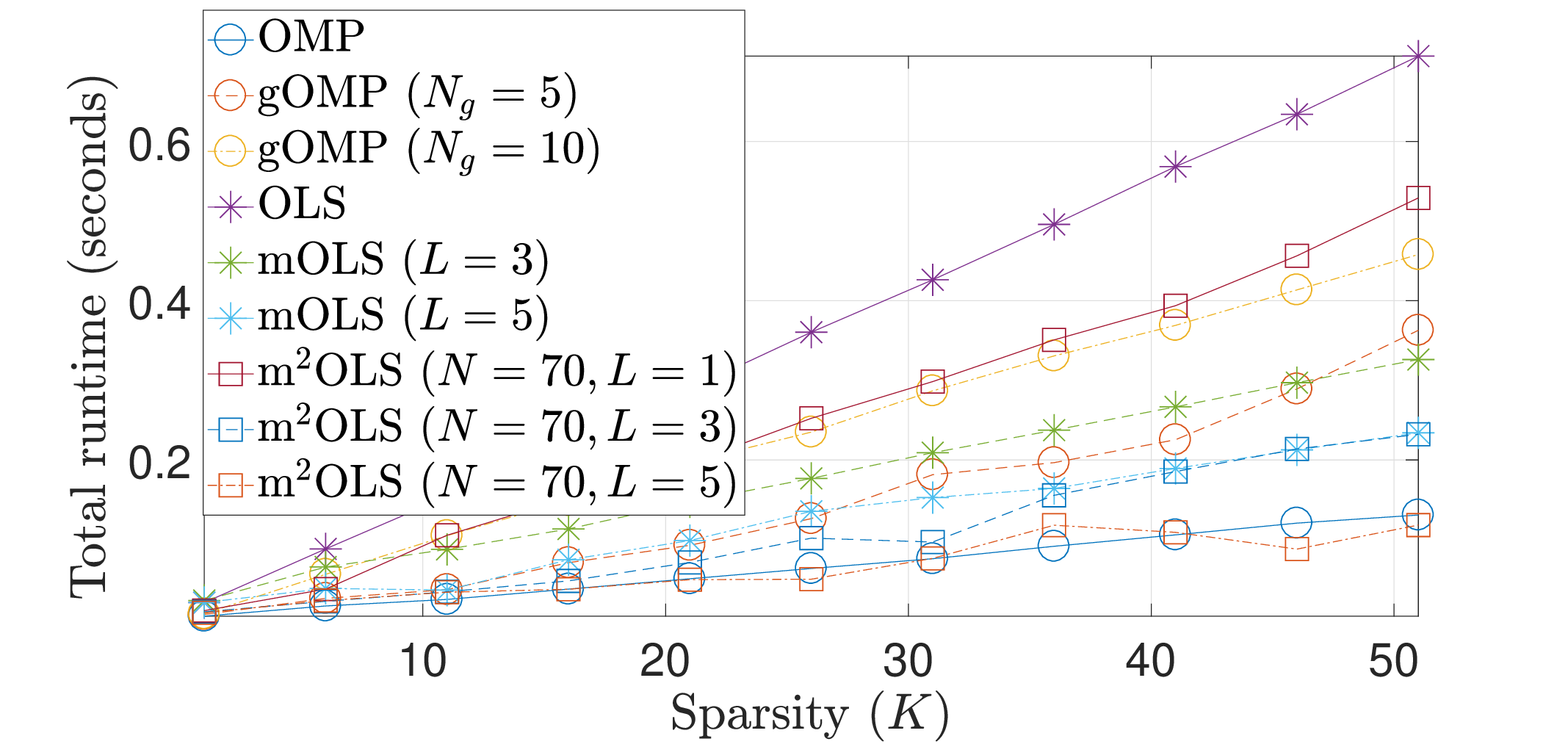}
\caption{$\tau=8$}
\label{fig:runtime_T=8}
\end{subfigure}
\caption{Mean runtime vs sparsity.} \label{fig:runtime_sparsity}
\end{figure}
\begin{figure}[t!]
\begin{subfigure}{.5\textwidth}
\centering
\includegraphics[height=1.5in,width=2.5in]{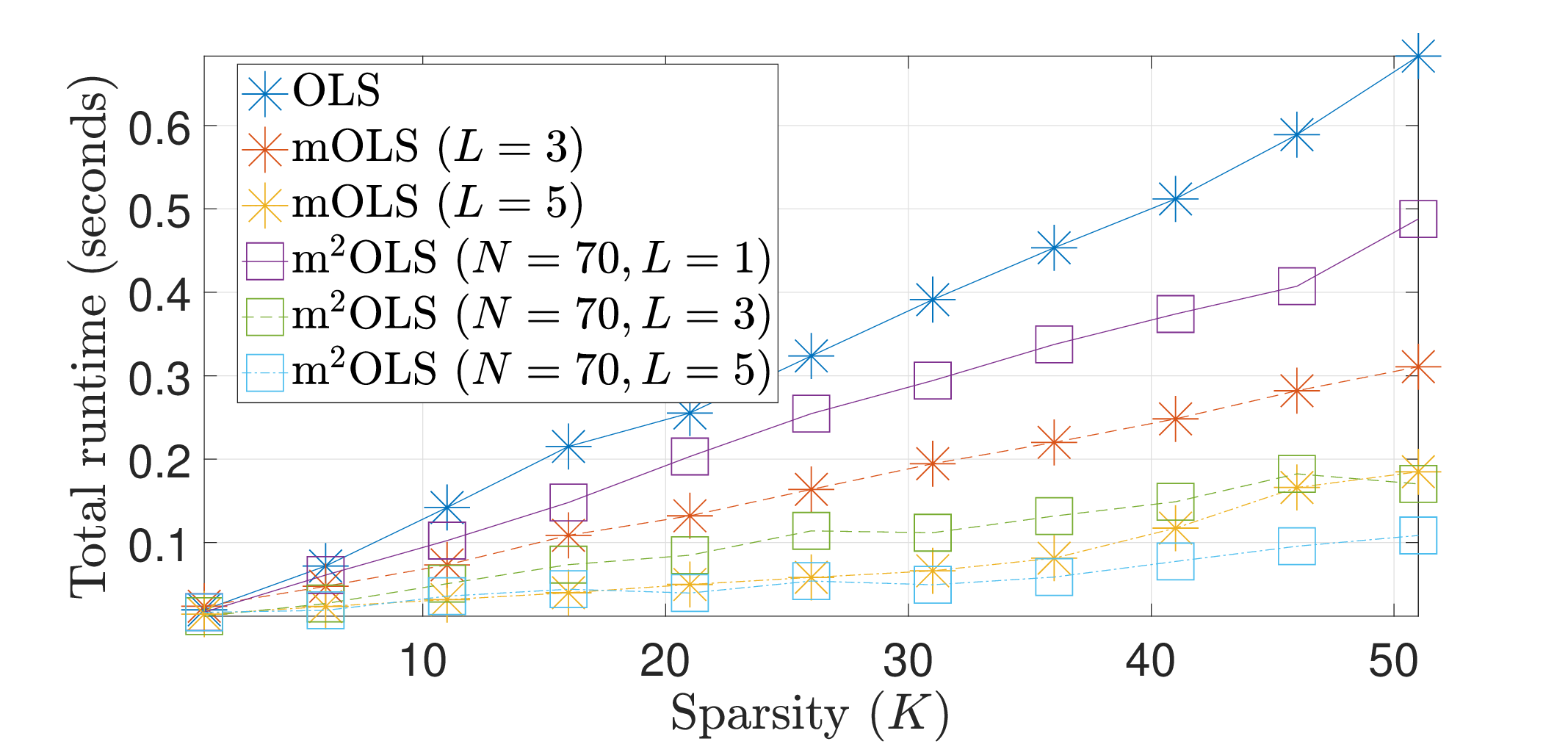}
\caption{$\tau=0$}
\label{fig:runtime_mols-m2ols_diff_N_T=0}
\end{subfigure}
\begin{subfigure}{.5\textwidth}
\centering
\includegraphics[height=1.5in,width=2.5in]{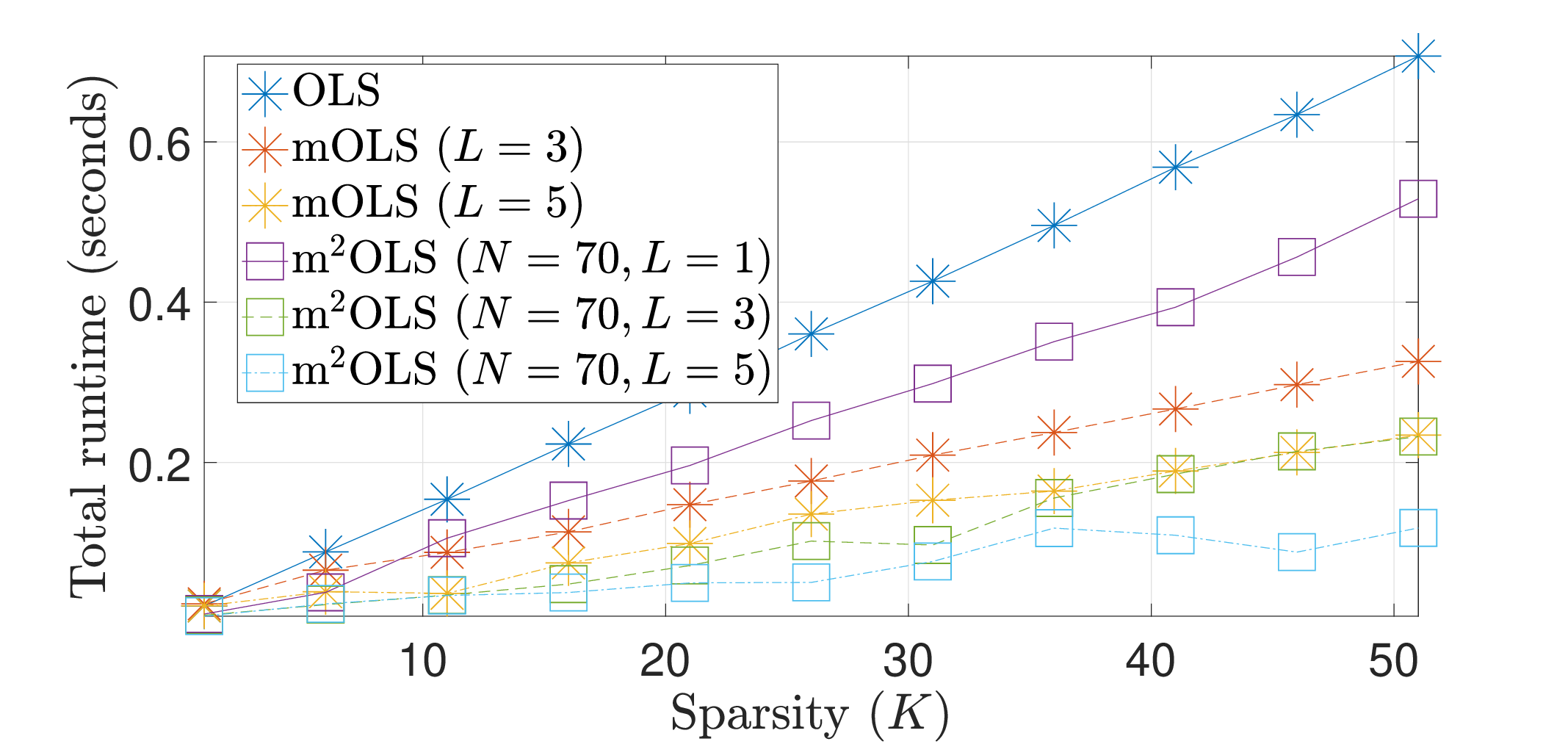}
\caption{$\tau=8$}
\label{fig:runtime_mols_m2ols_diff_N_T=8}
\end{subfigure}
\caption{Mean runtime vs sparsity (mOLS and m$^2$OLS) for different $L$.}
\label{fig:runtime_mols_m2ols_different_L}
\end{figure}
\begin{figure}[t!]
\begin{subfigure}{.5\textwidth}
\centering
\includegraphics[height=1.5in,width=2.5in]{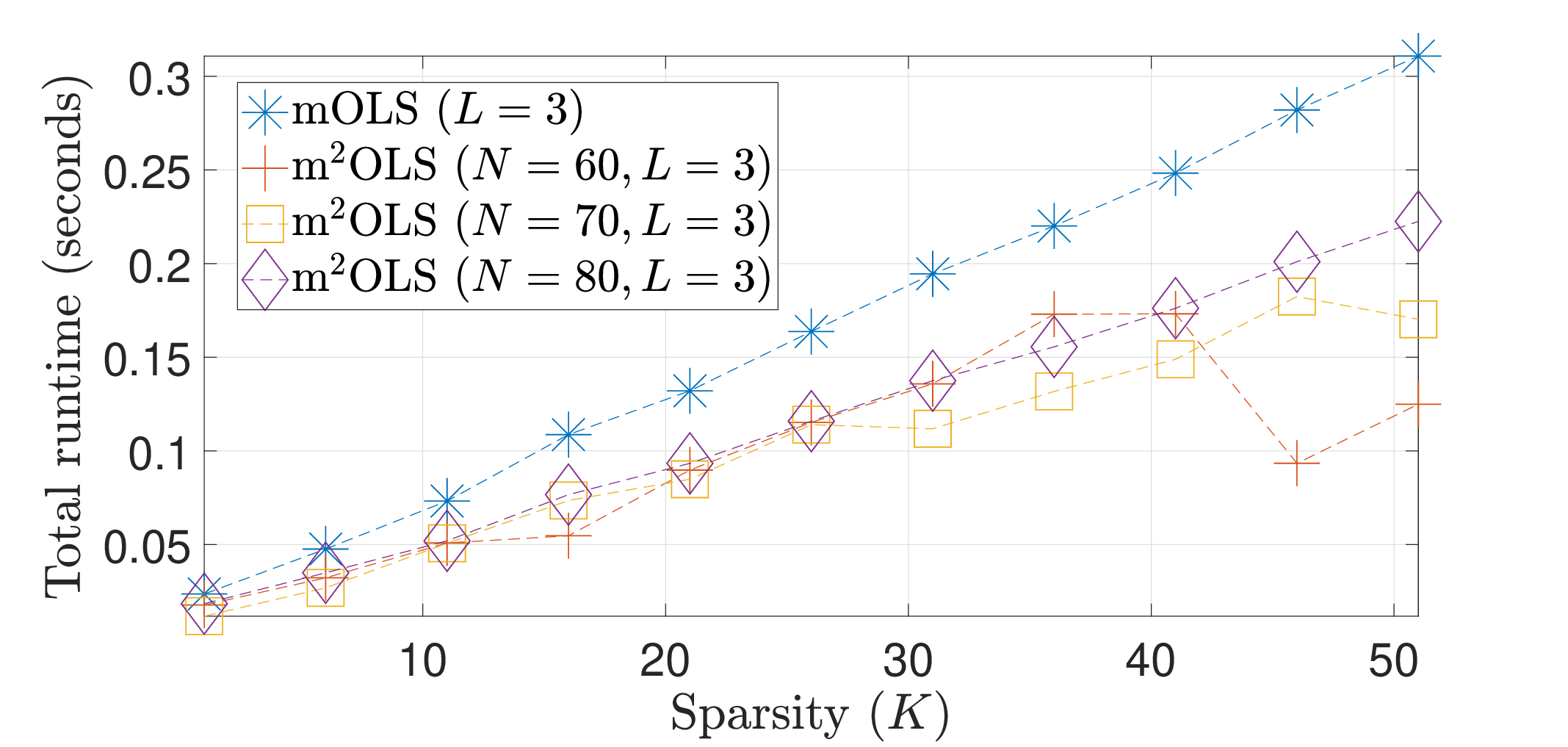}
\caption{$\tau=0$}
\label{fig:runtime_mols-m2ols_diff_L_T=0}
\end{subfigure}
\begin{subfigure}{.5\textwidth}
\centering
\includegraphics[height=1.5in,width=2.5in]{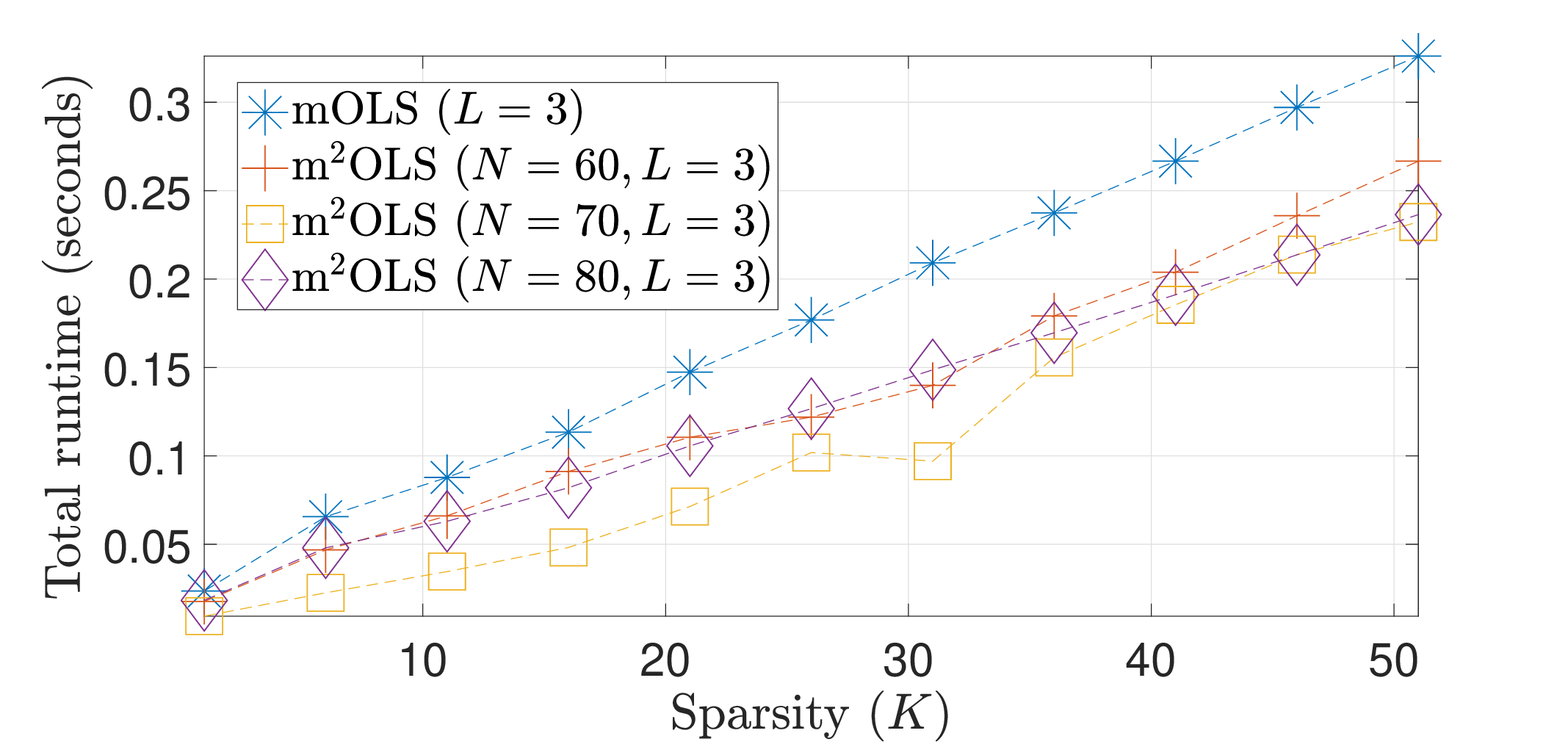}
\caption{$\tau=8$}
\label{fig:runtime_mols_m2ols_diff_L_T=8}
\end{subfigure}
\caption{Mean runtime vs sparsity (mOLS and m$^2$OLS) for different $N$.}
\label{fig:runtime_mols_m2ols_different_N}
\end{figure}
\begin{figure}[t!]
\begin{subfigure}{.5\textwidth}
\centering
\includegraphics[height=1.5in,width=2.5in]{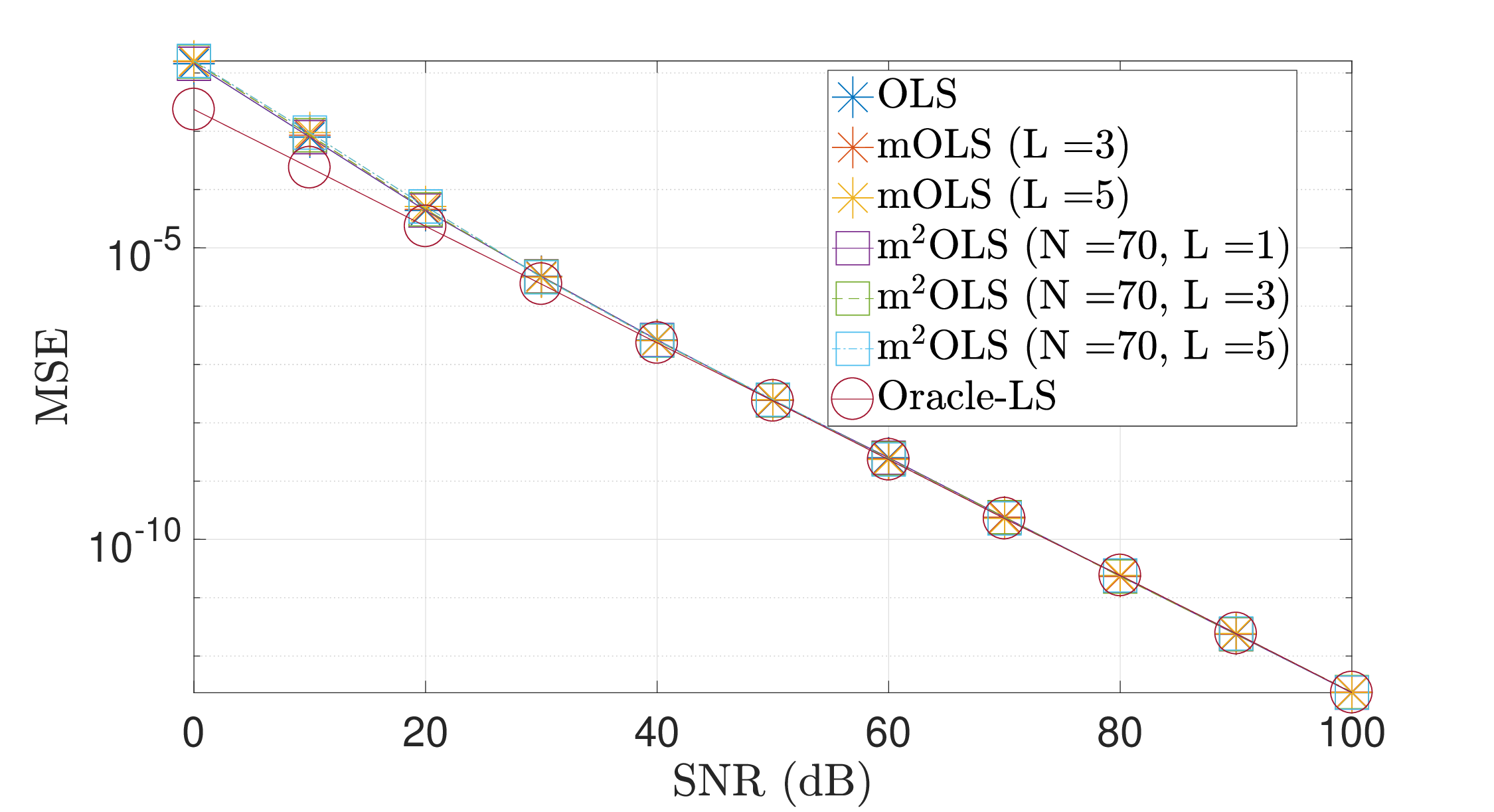}
\caption{$\tau=0$}
\label{fig:msd_vs_snr_T=0}
\end{subfigure}
\begin{subfigure}{.5\textwidth}
\centering
\includegraphics[height=1.5in,width=2.5in]{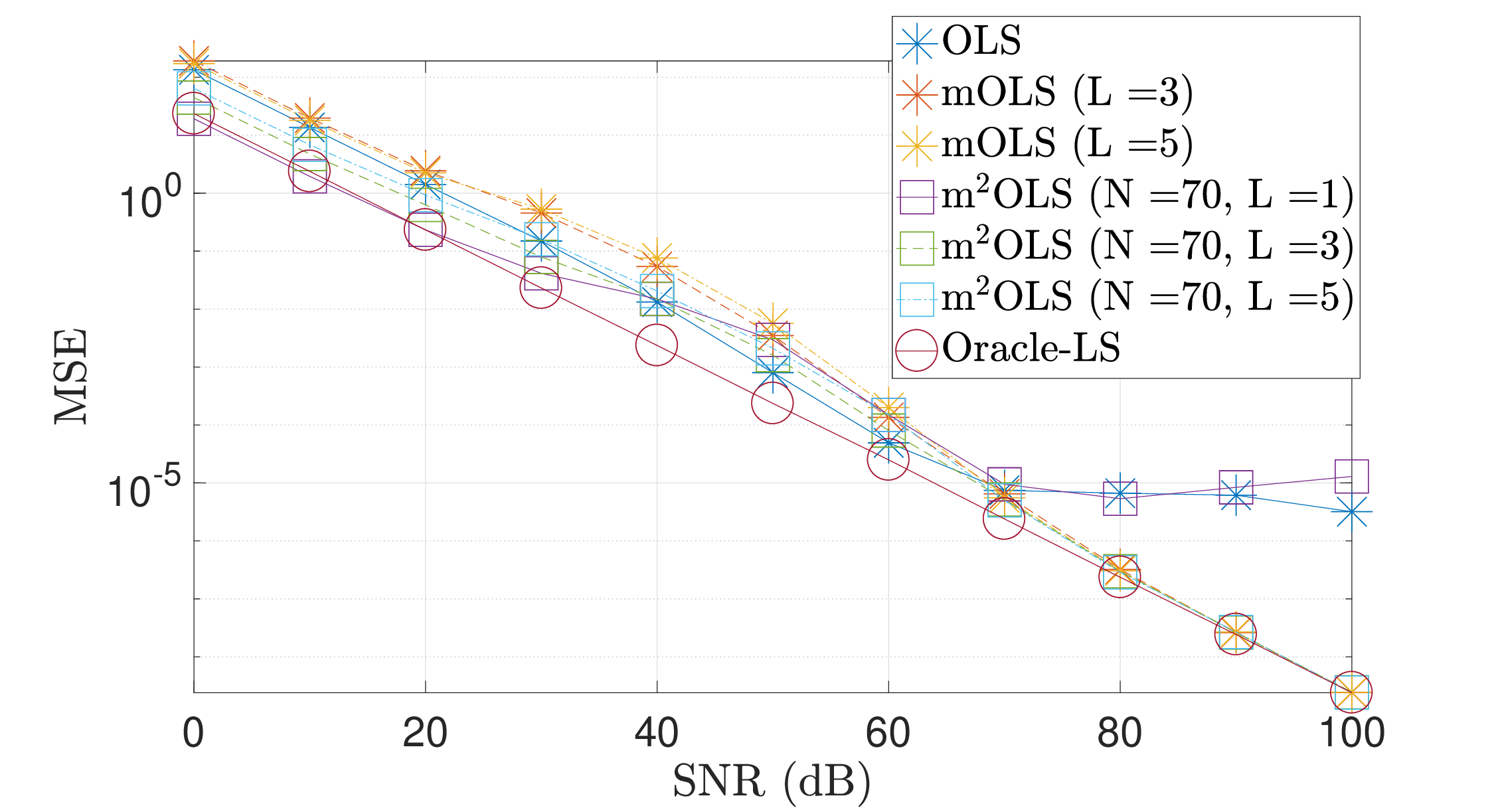}
\caption{$\tau=8$}
\label{fig:msd_vs_snr_T=8}
\end{subfigure}
\caption{Mean Square Error (MSE) vs SNR ($K=30$)}
\label{fig:msd_vs_snr}
\end{figure}
\begin{figure}[t!]
\begin{subfigure}{.5\textwidth}
\centering
\includegraphics[height=2in,width=3.5in]{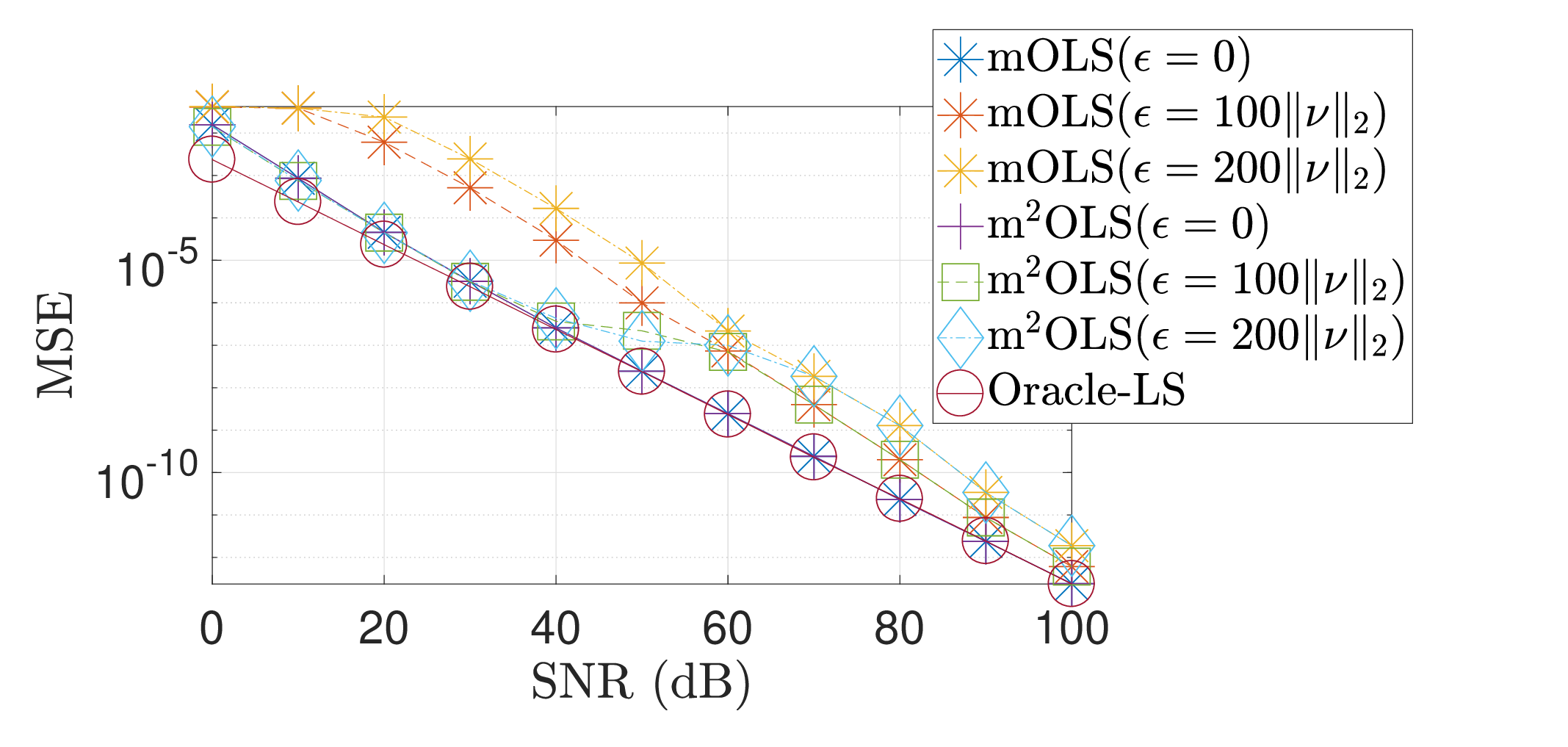}
\caption{$\tau=0$}
\label{fig:msd_vs_snr_eps_T=0}
\end{subfigure}
\begin{subfigure}{.5\textwidth}
\centering
\includegraphics[height=2in,width=3.5in]{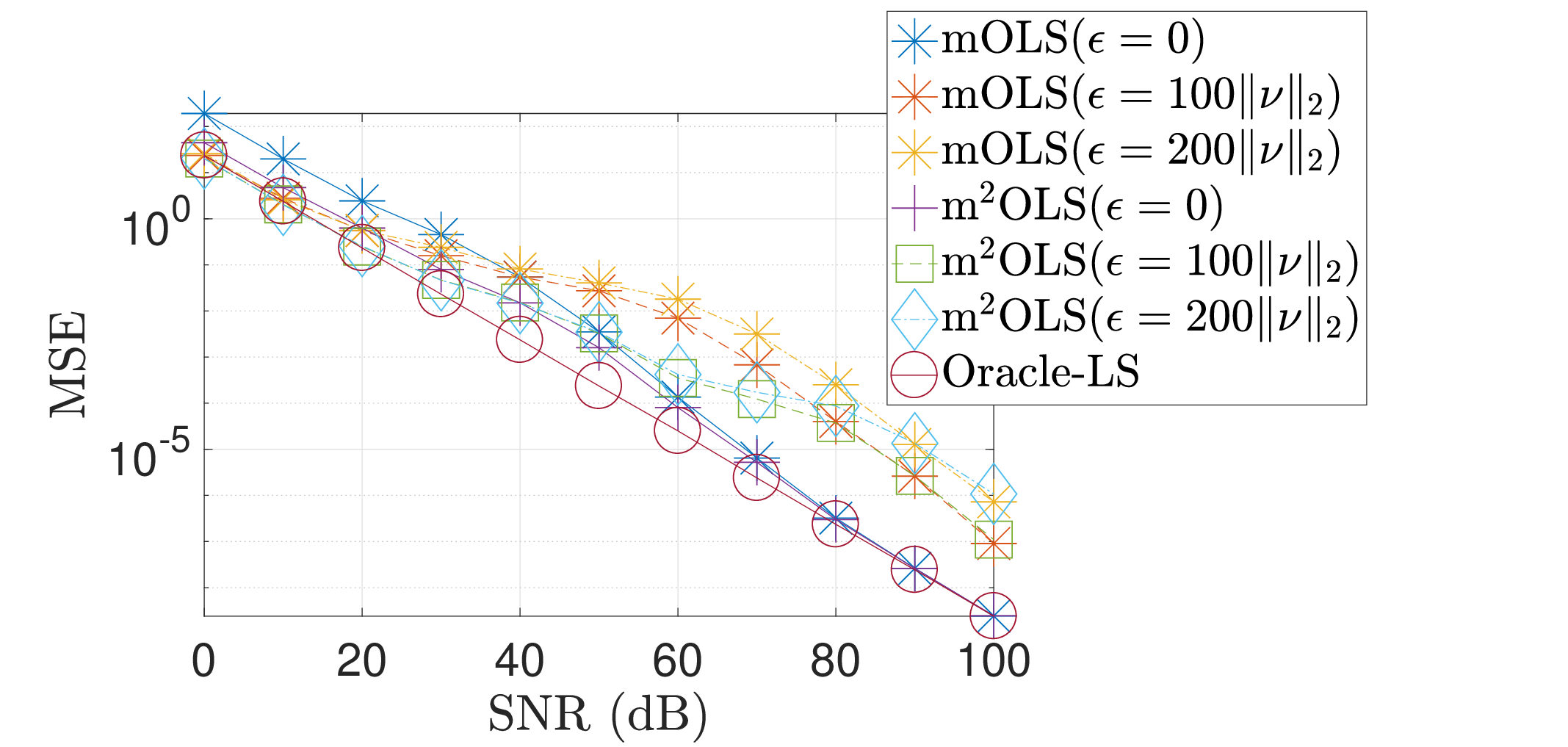}
\caption{$\tau=8$}
\label{fig:msd_vs_snr_eps_T=8}
\end{subfigure}
\caption{Mean Square Error (MSE) vs SNR ($K=30,\ N=70,\ L=3$) for different values of $\epsilon$}
\label{fig:msd_vs_snr}
\end{figure}

For simulation, we constructed measurement matrices with
correlated entries, as used by Soussen et
al~\cite{soussen2013joint}. For this, first a matrix $\bvec{A}$ is
formed such that $a_{ij}=[\bvec{A}]_{ij}$ is given by
$a_{ij}=n_{ij}+t_j$ where $n_{ij}\sim \mathcal{N}(0,1/m)$ i.i.d.
$\forall i,j$, $\ t_j\sim \mathcal{U}[0,\tau]\forall j$, and
$\{n_{ij}\}$ is statistically independent of $\{t_k\}$, $\forall
i,j,k$. The measurement matrix $\bvec{\Phi}$ is then constructed
from  $\bvec{A}$ as $\phi_{ij}=a_{ij}/\norm{\bvec{a}_j}$, where
$\phi_{ij}=[\bvec{\Phi}]_{ij}$ and $\bvec{a}_i$ denotes the $i$-th
column of $\bvec{A}$. Note that in the construction process for
$\bvec{\Phi}$, the random variables $n_{ij}$ play the role of
additive i.i.d. noise process, added to the elements of a rank $1$
matrix, with columns $\{t_i\bvec{1}\}_{i=1}^n$, where $\bvec{1}$
denotes a $m\times 1$ vector with all entries equal to one. If the
value of $\tau$ becomes large as compared to the variance $1/m$ of
$n_{ij}$, then the matrix $\bvec{\Phi}$ resembles a rank $1$
matrix with normalized columns. For all the simulations, the
values of $m,\ n$ were fixed at $500, 800$ respectively while the
sparsity $K$ was varied. The nonzero elements of $\bvec{x}$ were
drawn randomly from i.i.d Gaussian distribution and $\tau$ was
chosen to have values either $0$ or $8$. Note that higher the
value of $\tau$, more will be the correlation (taken as the
absolute value of the inner product, which is a measure of
\emph{coherence}) between the columns of $\bvec{\Phi}$. Thus,
$\tau = 0$ produces a matrix with uncorrelated columns while $\tau
= 8$ produces a matrix with reasonably correlated columns. Furthermore, different values for the window sizes for the preselection stage ($N$) of m$^2$OLS and the identification stages ($L$) of both mOLS and m$^2$OLS were used for the simulation. Particularly, the values $\{60,\ 70,\ 80\}$ were used for $N$, and $\{1,\ 3,\ 5\}$ were used for $L$. Moreover, we used different values for $N_g$, the number of indices identified at the identification stage of gOMP, from the set $\{1,\ 5, 10\}$.
%

For each value of $K$, the gOMP, mOLS and m$^2$OLS were run till they
converge or upto the $K$-th step of iteration, whichever is
earlier, and the experiment was conducted $500$ times. To evaluate
the performance of the algorithms, three performance metrics were
considered, namely, recovery probability, mean number of
iterations ($\le K$) for convergence and mean runtime. Of these
the recovery probability was obtained by counting the number of
times out of the $500$ trials each algorithm converges, while for the other two, averaging was
done over the $500$ trials~\cite{tropp2007signal}. 
%
%
 In the \textbf{first} simulation exercise, the recovery probabilities are plotted
against $K$. The plots,  shown in
Fig.~\ref{fig:recovery-prob_vs_K}(a) and (b) for $\tau=0$ and
$\tau=8$ respectively, suggest that even for highly correlated
dictionaries ($\tau=8$), the probability of recovery exhibited by
m$^2$OLS is identical to mOLS over the entire sparsity range
considered, and both of them outmatch the recovery probability performance of gOMP. However, it is observed from Fig.~\ref{fig:prob_K_T=8} that for the correlated dictionary the recovery probability of the gOMP does increase with the increase of $N_g$. The \textbf{second} simulation exercise evaluates the
average no. of iterations required by the two algorithms for exact
recovery for each value of $K$. The corresponding results, shown
in Fig.~\ref{fig:no_iteration_sparsity}(a) and (b) for $\tau=0$
and $\tau=8$ respectively, reveal that for the uncorrelated case
($\tau=0$), both mOLS and m$^2$OLS algorithms require the same average number of
iterations for successful recovery, and it is only under $\tau=8$
that as $K$ increases beyond a point, there is a marginal increase
in the average number of iterations in m$^2$OLS over mOLS. On the other hand gOMP requires relatively smaller number of iterations for $\tau =0$ and for $\tau = 8$ the required number of iterations increase with increase in $N_g$. In our \textbf{third} exercise, we evaluated the average of total runtime
for all three algorithms, against $K$. The corresponding results are 
shown in Fig.~\ref{fig:runtime_sparsity}(a) and (b) and the results for mOLS and m$^2$OLS are shown in Fig.~\ref{fig:runtime_mols_m2ols_different_L}, for $\tau=0$
and $\tau=8$ respectively. The figures demonstrate the superiority of the
proposed m$^2$OLS algorithm over mOLS as well as gOMP, as the former is seen to
require much less running time than mOLS for both values of
$\tau$, and has runtime closer to gOMP for $\tau =0$, and less than gOMP for $\tau = 8$. The results also illustrate that the runtime of m$^2$OLS decreases with increasing $L$ while maintaining lesser runtime than mOLS. This validates the conjectures made at the end of
Sec~\ref{sec:complexity-analysis} regarding the reduced
computational overhead of m$^2$OLS over a large range of sparsity
values. The plots also suggest that for $\tau = 8$, while gOMP with large $N_g$ can exhibit probability of recovery performance almost as good as mOLS (as demonstrated by Fig.~\ref{fig:prob_K_T=8} ), the runtime is significantly higher compared to m$^2$OLS. This is because as the plots in Fig.~\ref{fig:no_iteration_T=8} suggest, for larger $N_g$ gOMP requires much larger number of iterations for convergence than mOLS and m$^2$OLS even when $N_g$ is high, making its runtime higher. We also plot in Fig.~\ref{fig:runtime_mols_m2ols_different_N} runtimes for mOLS and m$^2$OLS for different $N$ with $L=3$ fixed. The results here demonstrate that the runtime of m$^2$OLS can be controlled by changing the preselection step size $N$, which is intuitively expected. In our \textbf{fourth} exercise, we ran the mOLS and
m$^2$OLS algorithms with measurements corrupted by additive
Gaussian noise with varying SNR (as defined in
Sec~\ref{sec:signal-recovery-theoretical-conditions}). The mean
square error (MSE) is computed as
$\norm{\hat{\bvec{x}}-\bvec{x}}^2$ where $\bvec{x}$ is the
original vector and $\hat{\bvec{x}}$ is its estimate as produced
by the algorithm. As a benchmark, the MSE of the Oracle estimator
is plotted, where the Oracle estimator computes the least squares
estimate of the optimal vector in presence of noise. For this experiment we consider $N=70,\ L=1,3,5$. The plots in Fig.~\ref{fig:msd_vs_snr}(a) and (b) demonstrate that for
uncorrelated dictionaries ($\tau =0$) the two algorithms exhibit
almost the same performance, while for correlated dictionaries ($\tau =8$),
m$^2$OLS has actually a slightly better MSE performance than mOLS
for different values of $L$. Also, the performance of mOLS as well as m$^2$OLS is for $L=1$ is the best in the low SNR region and the worst in the high SNR region. This is because when $L=1$, mOLS and m$^2$OLS incurs smaller error than for $L=3,5$, where larger error occurs because of selection of larger number of incorrect indices. We also experiment with the sensitivity of the MSE performances of the mOLS and m$^2$OLS algorithms with respect to the parameter $\epsilon$ on which the termination of the mOLS and m$^2$OLS algorithms depend on. We fix $N=70,\ L=3$, and take the values of $\epsilon$ as $\epsilon=\eta\norm{\bvec{\nu}}$, where $\bvec{\nu}$ is the measurement noise vector, and $\eta\in\{0,100,200\}$. We observe from the Fig.~\ref{fig:msd_vs_snr_eps_T=0} and \ref{fig:msd_vs_snr_eps_T=8} that for both the uncorrelated and correlated dictionaries, the MSE of m$^2$OLS is smaller than that of mOLS. Furthermore, early termination (high $\epsilon$) increases the MSE of m$^2$OLS in the high SNR region when $\tau=0$, whereas the MSE of mOLS increases in an even lower SNR region. Moreover, in the correlated dictionary,  early termination reduces MSE of both mOLS and m$^2$OLS in low SNR region, and increases it for high SNR region. 
\section{Conclusion}
In this paper we have proposed a greedy algorithm for sparse
signal recovery which preselects a few ($N$) possibly ``good''
indices according to correlation of the respective columns of the
measurement matrix with a residual vector, and then uses an mOLS
step to identify a subset of these indices (of size $L$) to be
included in the estimated support set. We have carried out a
theoretical analysis of the algorithm using RIP and have shown
that for the noiseless signal model, if the sensing matrix
satisfies the RIP condition
$\delta_{LK+N-L+1}<\frac{\sqrt{L}}{\sqrt{L}+\sqrt{L+K}}$, then the
m$^2$OLS algorithm is guaranteed to exactly recover a $K$ sparse
unknown vector, satisfying the measurement model, within $K$
steps. We further extended our analysis to a noisy measurement
setup and worked out bounds on the measurement SNR analytically,
which guarantees exact recovery of the support of the unknown
sparse vector within $K$ iterations. We have also presented the
computational steps of both mOLS and m$^2$OLS in a MGS based
efficient implementation and carried out a comparative analysis of
their computational complexities, which showed that m$^2$OLS
enjoys significantly reduced computational overhead compared to
mOLS, especially for large $n$ and / or small $K$. Finally,
through numerical simulations, we have verified that the
introduction of the preselection step indeed leads to less
computation time, and that the recovery performance of m$^2$OLS in
terms of recovery probability and number of iterations for success
is highly competitive with mOLS for a wide range of parameter
values.
\appendix
\section{Proofs of Theorem~\ref{thm:noiseless-recovery} and Theorem~\ref{thm:noisy-recovery}}
\label{sec:appendix-proof-thms}
%
%
\subsubsection{Success at the first iteration}
At the first iteration, the conditions for success are $S^1\cap T\ne \emptyset$, and $ h^1\cap T\ne \emptyset$.
In order to have these satisfied, we first observe the following:
\begin{lem}
\label{lem:lemma-observations-average-energy}
\begin{align}
\label{eq:obs1}
\frac{1}{\sqrt{N}}\|\bvec{\Phi}_{S^1}^t\bvec{y}\|_2 & \ge \frac{1}{\sqrt{K}}\|\bvec{\Phi}_{T}^t\bvec{y}\|_2,\quad (N\le K)\\
\label{eq:obs2}
\|\bvec{\Phi}_{S^1}^t\bvec{y}\|_2 & \ge \|\bvec{\Phi}_{T}^t\bvec{y}\|_2,\quad (N> K)\\
\label{eq:obs3}
\frac{1}{\sqrt{L}}\|\bvec{\Phi}_{T^1}^t\bvec{y}\|_2 & \ge \frac{1}{\sqrt{K}}\|\bvec{\Phi}_{T}^t\bvec{y}\|_2.
\end{align}
\end{lem}
\begin{proof}
The proof is given in Appendix~\ref{sec:appendix-proof-lemma-observation-average-energy}.
\end{proof}
Now, from~\eqref{eq:obs1} and \eqref{eq:obs2} above, \begin{align*}
\norm{\bvec{\Phi}_{S^1}^t\bvec{y}} & \ge \min\left\{1,\sqrt{\frac{N}{K}}\right\}\norm{\bvec{\Phi}_T^t\bvec{y}}\\
\ & =\min\left\{1,\sqrt{\frac{N}{K}}\right\}\norm{\bvec{\Phi}_T^t\bvec{\Phi}_T\bvec{x}_T+\bvec{\Phi}_T^t\bvec{e}}\\
\ & \ge\min\left\{1,\sqrt{\frac{N}{K}}\right\}\left[\norm{\bvec{\Phi}_T^t\bvec{\Phi}_T\bvec{x}_T}-\norm{\bvec{\Phi}_T^t\bvec{e}}\right]\\
\ & \stackrel{(a)}{\ge}\min\left\{1,\sqrt{\frac{N}{K}}\right\}\left[(1-\delta_K)\norm{\bvec{x}}-\sqrt{1+\delta_K}\norm{\bvec{e}}\right],
\end{align*}
where the inequalities in step $(a)$ follow from Lemmas~\ref{lem:max-min-eigenvalues-phi_t*phi} and \ref{lem:upper-bound-phi^t-u}, respectively.
If $S^{1}\cap T=\emptyset$, then, \begin{align*}
\norm{\bvec{\Phi}_{S^1}^t\bvec{y}} &=\norm{\bvec{\Phi}_{S^1}^t\bvec{\Phi}_T\bvec{x}_T+\bvec{\Phi}_{S^1}^t\bvec{e}}\\
\ &\stackrel{(b)}{\le} \delta_{N+K}\norm{\bvec{x}}+\sqrt{1+\delta_N}\norm{\bvec{e}},
\end{align*}
where the inequalities in step $(b)$ follow from
Lemmas~\ref{lem:orthogonal-sparse1},
and~\ref{lem:upper-bound-phi^t-u}, respectively. Hence $S^1\cap
T\ne \emptyset$ is guaranteed if \begin{align}
\delta_{N+K}\norm{\bvec{x}}+\sqrt{1+\delta_N}\norm{\bvec{e}} & < \min\left\{1,\sqrt{\frac{N}{K}}\right\} \left[(1-\delta_K)\norm{\bvec{x}}-\sqrt{1+\delta_K}\norm{\bvec{e}}\right].
\label{eq:noisy-start-condition1}
\end{align}
%
Again, in a similar manner as above, \begin{align*}
\norm{\bvec{\Phi}^t_{T^1}\bvec{y}} & \ge \sqrt{\frac{L}{K}}\norm{\bvec{\Phi}^t_{T}\bvec{y}} = \sqrt{\frac{L}{K}}\norm{\bvec{\Phi}^t_T\bvec{\Phi}_T \bvec{x}_T + \bvec{\Phi}_T^t\bvec{e}}\\
\ & \ge \sqrt{\frac{L}{K}}\left[(1-\delta_K)\norm{\bvec{x}}-\sqrt{1+\delta_K}\norm{\bvec{e}}\right].
\end{align*}
If $T^1\cap T=\emptyset$, we have \begin{align}
\norm{\bvec{\Phi}_{T^1}^t\bvec{y}} &=\norm{\bvec{\Phi}_{T^1}^t\bvec{\Phi}_T\bvec{x}_T+\bvec{\Phi}_{T^1}^t\bvec{e}}\nonumber\\
\label{eq:noisy-start-condition2}
\ &\le \delta_{L+K}\norm{\bvec{x}}+\sqrt{1+\delta_L}\norm{\bvec{e}}.
\end{align}
Hence, given that $S^1\cap T\ne \emptyset$, $T^1\cap T\ne \emptyset$ is guaranteed, if  \begin{align}
\label{eq:noisy-start-condition3}
\delta_{L+K}\norm{\bvec{x}}+\sqrt{1+\delta_L}\norm{\bvec{e}} & <\sqrt{\frac{L}{K}}\left[(1-\delta_K)\norm{\bvec{x}}-\sqrt{1+\delta_K}\norm{\bvec{e}}\right].
\end{align}

Since, $N\ge L$ and $K\ge L$ (by assumption), we have
$\delta_{N+K}\ge \delta_{L+K}$, and $\sqrt{\frac{L}{K}}\le
\min\left\{1,\sqrt{\frac{N}{K}}\right\}$. Therefore, a sufficient
condition for simultaneous satisfaction of
~\eqref{eq:noisy-start-condition1} and
~\eqref{eq:noisy-start-condition3} (i.e., for success at first
iteration) can be stated as follows:
\begin{align*}
\delta_{N+K}\norm{\bvec{x}}+\sqrt{1+\delta_N}\norm{\bvec{e}} & <\sqrt{\frac{L}{K}}\left[(1-\delta_K)\norm{\bvec{x}}-\sqrt{1+\delta_K}\norm{\bvec{e}}\right],
\end{align*}
or, equivalently,
\begin{align}
\label{eq:step1-equivalent-condition}
\Leftrightarrow \norm{\bvec{x}}\left(\sqrt{L}(1-\delta_K)-\sqrt{K}\delta_{N+K}\right) & >\norm{\bvec{e}}\left(\sqrt{K}\sqrt{1+\delta_N}+\sqrt{L}\sqrt{1+\delta_K}\right).
\end{align}
Note that as the RHS of \eqref{eq:step1-equivalent-condition} is positive, satisfaction of the above first requires the LHS to be positive.
\begin{itemize}
\item Noiseless case: For this, we have $\bvec{e=0}$. The inequality~\eqref{eq:step1-equivalent-condition} then leads to   \begin{align*}
\delta_{N+K}<\sqrt{\frac{L}{K}}(1-\delta_K).
\end{align*}
Since, $\delta_K< \delta_{N+K}$, the above is satisfied if the following condition holds:
\begin{align}
\delta_{N+K} & <\sqrt{\frac{L}{K}}(1-\delta_{K+N})\nonumber\\
\Leftrightarrow \delta_{N+K} & <\sqrt{L}/\left(\sqrt{L}+\sqrt{K}\right).
\label{eq:recovery-condition-first-step-noiseless}
\end{align}
\item Noisy case (i.e. $\norm{\bvec{e}}>\bvec{0}$): For this,
first the LHS of ~\eqref{eq:step1-equivalent-condition} must be
positive which is guaranteed
under~\eqref{eq:recovery-condition-first-step-noiseless}. Subject
to this, we need to condition the ratio
$\frac{\norm{\bvec{x}}}{\norm{\bvec{e}}}$ appropriately so
that~\eqref{eq:step1-equivalent-condition} is satisfied. Note that
since $\delta_{N+K}\ge \max\{\delta_N,\ \delta_K\}$,
~\eqref{eq:step1-equivalent-condition} is ensured under the
following condition:
 \begin{align*}
\norm{\bvec{x}}\left(\sqrt{L}(1-\delta_{N+K})-\sqrt{K}\delta_{N+K}\right) & >\norm{\bvec{e}}\sqrt{1+\delta_{N+K}}\left(\sqrt{K}+\sqrt{L}\right).
\end{align*}
 The above leads to the following condition on $\frac{\norm{\bvec{x}}}{\norm{\bvec{e}}}$
 for the first iteration to be successful under noisy observation

 \begin{align}
\label{eq:recovery-condition-first-step-noisy}
\frac{\norm{\bvec{x}}}{\norm{\bvec{e}}} & >\frac{\sqrt{1+\delta_{N+K}}(\sqrt{L}+\sqrt{K})}{\sqrt{L}-(\sqrt{L}+\sqrt{K})\delta_{N+K}}.
\end{align}
\end{itemize}

\subsubsection{Success at $(k+1)^{\mathrm{th}}$ iteration}
We assume that in each of the previous $k$ ($k<K$) iterations, at
least one correct index was selected, meaning, if $|T\cap
T^k|=c_k$, then $c_k\ge k$. Let $c_k<K$. Also define
$m_k:=|S^k\cap T\setminus T^k|,\ k\ge 1$, meaning, $m_i\ge 1,\ 1\le i\le k$.
For success of the $(k+1)^{th}$ iteration, we require $S^{k+1}
\cap T\setminus T^k\ne \emptyset$, and $ h^{k+1} \cap T\setminus T^k\ne \emptyset$
simultaneously, as this will ensure selection of at least one new true index at the $(k+1)$-th iteration.\\
\textbf{Condition to ensure $S^{k+1} \cap T\setminus T^k \ne
\emptyset$} : First consider the set $\support\setminus
(T\setminus T^k)$. If $|\support\setminus (T\setminus T^k)|< N$,
then, the condition $S^{k+1}\cap T\setminus T^k\ne \emptyset$ is
satisfied trivially. We therefore consider cases where
$\abs{\support\setminus (T\setminus T^k)}\ge N$, for which we
define the following:
\begin{itemize}
\item $W^{k+1}:=\displaystyle\argmax_{S\subset \support\setminus
(T\setminus T^k):\ |S|=N}\norm{\bvec{\Phi}_S^t\bvec{r}^k}$. \item
$\alpha_N^k:=\min_{i\in
W^{k+1}}\abs{\inprod{\bvec{\phi}_i}{\bvec{r}^k}}$. \item
$\beta_1^k:=\max_{i\in T\setminus
T^k}\abs{\inprod{\bvec{\phi}_i}{\bvec{r}^k}}$.
\end{itemize}
Clearly, $S^{k+1}\cap T\setminus T^k\ne \emptyset$, if $\beta_1^k>\alpha_N^k$.
It is easy to see that \begin{align*} {\alpha_N^k} & \le
\frac{\norm{\bvec{\Phi}_{W^{k+1}}^t\bvec{r}^k}}{\sqrt{N}}=\frac{\norm{\bvec{\Phi}_{W^{k+1}\setminus
T^k}^t\bvec{r}^k}}{\sqrt{N}},
\end{align*}
since $\bvec{r}^k$ is orthogonal to the columns of
$\bvec{\Phi}_{T^k}$. Now, using the derivation of ~\cite[Eq.(9)]{satpathi2013improving}, it is straightforward to derive that \begin{align}
\label{eq:rk_expression}
\bvec{r}^{k}&=\bvec{\Phi}_{T\cup T^k}\bvec{x}'_{T\cup
T^k}+\dualproj{T^k}\bvec{e},
\end{align}
where we have expressed the projection
$\proj{T^k}\bvec{\Phi}_{T\setminus T^k}\bvec{x}_{T\setminus T^k}$
as a linear combination of the columns of $\bvec{\Phi}_{T^k}$,
i.e., as $\bvec{\Phi}_{T^k}\bvec{u}_{T^k}$ for some
$\bvec{u}_{T^k}\in \real^{Lk}$, and,
 \begin{align*}
\bvec{x}'_{T\cup T^k}=\begin{bmatrix}
\bvec{x}_{T\setminus T^k} \\
-\bvec{u}_{T^k}
\end{bmatrix}.
\end{align*}
Then, using the expression for $\bvec{r}^k$ from Eq.~\eqref{eq:rk_expression} and using steps similar to the analysis for~\cite[Eq.(26)]{li2015sufficient}, it follows that 
\begin{align}
\label{eq:alpha-inequality}\alpha_N^k \le & \frac{1}{\sqrt{N}}\left(\delta_{N+Lk+K-c_k}\norm{\bvec{x}'_{T\cup T^k}}+\sqrt{1+\delta_{N}}\norm{\bvec{e}}\right).
\end{align}
On the other hand, \begin{align*}
\beta_1^k & \ge \frac{1}{\sqrt{K-c_k}}\norm{\bvec{\Phi}_{T\setminus T^k}^t\bvec{r}^k}.
\end{align*}
 Again, using the expression for $\bvec{r}^k$ from Eq.~\eqref{eq:rk_expression} and steps similar to the analysis for~\cite[Eq.(25)]{li2015sufficient} it follows that, \begin{align}
\beta_1^k  \ge & \frac{(1-\delta_{Lk+K-c_k})\norm{\bvec{x}'_{T\cup T^k}} - \sqrt{1+\delta_{Lk+K-c_k}}\norm{\bvec{e}}}{\sqrt{K-c_k}}
\label{eq:beta-inequality} 
 \end{align}
 Then, from \eqref{eq:alpha-inequality} and \eqref{eq:beta-inequality}, it follows that $S^{k+1}\cap T\ne \emptyset$ if \begin{align}
\frac{\left((1-\delta_{LK-L+1})\norm{\bvec{x}'_{T\cup T^k}}-\sqrt{1+\delta_{LK-L+1}}\norm{\bvec{e}}\right)}{\sqrt{K-c_k}} & > \frac{\left(\delta_{N+LK-L+1}\norm{\bvec{x}'_{T\cup T^k}}+\sqrt{1+\delta_{N}}\norm{\bvec{e}}\right)}{\sqrt{N}}.
 \label{eq:alpha-vs-beta-inequality-recovery}
 \end{align}
 where we have used the fact that $1\le k\le c_k$  and $k\le K-1$, implying, $Lk+K-c_k\le
(L-1)k+K\le (L-1)(K-1)+K=LK-L+1$, and the monotonicity of the RIC.

\textbf{Condition to ensure $h^{k+1} \cap T\setminus T^k\ne
\emptyset$} : First consider the set $S^{k+1}\setminus (T\setminus
T^k)$. If $|S^{k+1}\setminus (T\setminus T^k)|<L$, then the
condition $h^{k+1} \cap T\setminus T^k\ne \emptyset$ is satisfied
trivially. Therefore, we consider cases where $|S^{k+1}\setminus
(T\setminus T^k)|\ge L$. Then, using the definition of $a_i,\;i\in
S^{k+1}$ as given in Lemma~\ref{lem:identification}, we define the
following :
 \begin{itemize}
 \item $\displaystyle V^{k+1}=\argmax_{S\subset S^{k+1}\setminus (T\setminus T^k):|S|=L}\sum_{i\in S}a_i$.
 \item $\displaystyle u_1^k:=\max_{i\in S^{k+1}\cap T\setminus T^k}a_i\equiv \max_{i\in S^{k+1}\cap T}a_i$.
 \item $\displaystyle v_L^k=\min_{i\in V^{k+1}}a_i$.
 \end{itemize}
 From Lemma~\ref{lem:identification}, $u_1^k>v_L^k$ will ensure $h^{k+1}\cap T\setminus T^k\ne \emptyset$.
 Now, \begin{align*}
 u_1^k &=\max_{i\in S^{k+1}\cap T}a_i=\max_{i\in (S^{k+1}\cap T)\setminus \tilde{T}^K}a_i\\
 \ &\ge \max_{i\in (S^{k+1}\cap T)\setminus \tilde{T}^K}\abs{\inprod{\bvec{\phi}_i}{\bvec{\bvec{r}^k}}}~(\mathrm{since}\ \norm{\dualproj{T^k}\bvec{\phi_i}}\le\norm{\bvec{\phi_i}}=1)\\
 \ & \ge \max_{i\in T}\abs{\inprod{\bvec{\phi}_i}{\bvec{\bvec{r}^k}}} ~(\mathrm{from\; the\; definition\; of\;}\   S^{k+1}\ \mathrm{and}\ \tilde{T}^K)\\
 \ &\ge \frac{\norm{\bvec{\Phi}_{T\setminus T^k}^t\bvec{r}^k}}{\sqrt{K-c_k}}~(\mathrm{since} \inprod{\bvec{\phi}_i}{\bvec{\bvec{r}^k}}=0\; \mathrm{for}\; i\in T^k).
 \end{align*}
Now, recalling that $ \bvec{r}^k=\bvec{\Phi}_{T\cup
T^k}\bvec{x}'_{T\cup T^k}+\dualproj{T^k}\bvec{e}$ and using steps similar to those used for \cite[Eq.(25)]{li2015sufficient}
we have,
\begin{align*}
\norm{\bvec{\Phi}_{T\setminus T^k}^t\bvec{r}^k} & \ge (1-\delta_{Lk+K-c_k})\norm{\bvec{x}'_{T\cup T^k}}-\sqrt{1+\delta_{Lk+K-c_k}}\norm{\bvec{e}}\\
 \ & \ge (1-\delta_{LK-L+1})\norm{\bvec{x}'_{T\cup T^k}}-\sqrt{1+\delta_{LK-L+1}}\norm{\bvec{e}},
\end{align*}
 Thus, \begin{align}
 \label{eq:u_1-lower-bound}
 u_1^k & \ge \frac{1}{\sqrt{K-c_k}}\left[(1-\delta_{LK-L+1})\norm{\bvec{x}'_{T\cup T^k}}-\sqrt{1+\delta_{LK-L+1}}\norm{\bvec{e}}\right].
 \end{align}
 On the other hand \begin{align}
v_L^k=&  \min_{i\in V^{k+1}}a_i\nonumber\\
\ \le & \frac{1}{\sqrt{L}}\sqrt{\sum_{i\in V^{k+1}}a_i^2}\nonumber\\
\ \le & \frac{1}{\sqrt{L}}\sqrt{\sum_{i\in S^{k+1}\setminus (T\setminus T^k)}a_i^2}\nonumber \quad (\because V^{k+1}\subset S^{k+1}\setminus (T\setminus T^k))\\
\ = & \frac{1}{\sqrt{L}}\sqrt{\sum_{i\in S^{k+1}\setminus T}a_i^2}\nonumber \\
\label{eq:v_L-intermediate-inequality} \ \le &
\frac{\frac{1}{\sqrt{L}}\norm{\bvec{\Phi}_{S^{k+1}\setminus
T}^t\bvec{r}^k}}{\min_{i\in S^{k+1}\setminus (T\cup
\tilde{T}^K)}\|\dualproj{T^k}\bvec{\phi}_i\|_2}.
\end{align}
Now, $\bvec{\phi}_i\; \forall i\in H$ can be written as
$\bvec{\Phi}\bvec{\nu}_i$, where $\bvec{\nu}_i$ is the $i$-th
column of the $n\times n$ identity matrix. Then, noting that
$supp(\bvec{\nu}_i)=\{i\}$ with $|\{i\}|=1$, for $i\in
S^{k+1}\setminus (T\cup \tilde{T}^K)$,
\begin{align}
\norm{\dualproj{T^k}\bvec{\phi}_i}^2 & =\norm{\bvec{A}_{T^k}\bvec{\nu}_i}^2\nonumber\\
\ &\stackrel{\mathrm{Lemma}~\ref{lem:orthogonal-sparse2}}{\ge} \left(1-\left(\frac{\delta_{Lk+1}}{1-\delta_{Lk+1}}\right)^2\right)\norm{\bvec{\phi}_i}^2\nonumber\\
\label{eq:A_Tk-inequlity} \ &\ge
\left(1-\left(\frac{\delta_{LK-L+1}}{1-\delta_{LK-L+1}}\right)^2\right),
\end{align}
since, $\norm{\bvec{\phi}_i}=1$ and $k\le K-1$ (note that
application of Lemma 3.5 requires $\delta_1<1$, which is trivially
satisfied by the proposed sufficient condition
\eqref{eq:final-condition2}). Also, using the expression for $\bvec{r}^k$ from Eq.~\eqref{eq:rk_expression} and using steps similar to the ones used to obtain~\cite[Eq.(26)]{li2015sufficient}, we obtain
\begin{align*}
 \norm{\bvec{\Phi}_{S^{k+1}\setminus T}^t\bvec{r}^k} &\le \delta_{Lk+K+N-m_{k+1}-c_k}\norm{\bvec{x}'_{T\cup T^k}}+\sqrt{1+\delta_{N-m_{k+1}}}\norm{\bvec{e}}\\
 \ &\le \delta_{N+LK-L+1}\norm{\bvec{x}'_{T\cup T^k}}+\sqrt{1+\delta_{N+LK-L+1}}\norm{\bvec{e}},
 \end{align*}
 Then, noting that $\delta_{LK-L+1}<\delta_{LK+N-L+1}$,  \begin{align}
 \label{eq:v_L-inequality}
 v_L^k < \frac{\delta_{LK+N-L+1}\norm{\bvec{x}'_{T\cup T^k}}+\sqrt{1+\delta_{LK+N-L+1}}\norm{\bvec{e}}}{\sqrt{L\left(1-\left(\frac{\delta_{LK+N-L+1}}{1-\delta_{LK+N-L+1}}\right)^2\right)}}.
 \end{align}
In order to ensure that the denominator of the RHS of above
remains real, we need $\delta_{LK+N-L+1}<1/2$. This is seen to be
satisfied trivially by the proposed sufficient condition
\eqref{eq:final-condition2}. For brevity, let us also denote
$LK+N-L+1$ by $R$.

 From Eq.~\eqref{eq:u_1-lower-bound}, and Eq.~\eqref{eq:v_L-inequality}, a sufficient condition
 to ensure $h^{k+1}\cap T\ne \emptyset$ is given by \begin{align}
 {\frac{1}{\sqrt{K-c_k}}\left[(1-\delta_{R})\norm{\bvec{x}'_{T\cup T^k}}-\sqrt{1+\delta_{R}}\norm{\bvec{e}}\right]} & \ge\frac{\delta_{R}\norm{\bvec{x}'_{T\cup
T^k}}+\sqrt{1+\delta_{R}}\norm{\bvec{e}}}{\sqrt{L\left(1-\left(\frac{\delta_{R}}{1-\delta_{R}}\right)^2\right)}}.
\label{eq:suff-cond-kth-step-temp1}
\end{align}
Thus, from Eq~\eqref{eq:alpha-vs-beta-inequality-recovery} and
Eq~\eqref{eq:suff-cond-kth-step-temp1}, a sufficient condition for
success at the $(k+1)^{\mathrm{th}}$ iteration will be as follows
:
\begin{align}
{\frac{1}{\sqrt{K-c_k}}\left[(1-\delta_{R})\norm{\bvec{x}'_{T\cup T^k}}-\sqrt{1+\delta_{R}}\norm{\bvec{e}}\right]} & \ge \max\left\{\frac{1}{\sqrt{N}},\frac{1}{\sqrt{L\left(1-\left(\frac{\delta_{R}}{1-\delta_{R}}\right)^2\right)}}\right\}\nonumber\\
\ & \times \left(\delta_{R}\norm{\bvec{x}'_{T\cup
T^k}}+\sqrt{1+\delta_{R}}\norm{\bvec{e}}\right).
\end{align}
Since
$L\left(1-\left(\frac{\delta_{R}}{1-\delta_{R}}\right)^2\right)<L\le
N$, the above sufficient condition for success at the $k+1$-th
step boils down to the following :
\begin{align}
{\frac{1}{\sqrt{K-c_k}}\left[(1-\delta_{R})\norm{\bvec{x}'_{T\cup T^k}}-\sqrt{1+\delta_{R}}\norm{\bvec{e}}\right]} & \ge \frac{\delta_{R}\norm{\bvec{x}'_{T\cup
T^k}}+\sqrt{1+\delta_{R}}\norm{\bvec{e}}}{\sqrt{L\left(1-\left(\frac{\delta_{R}}{1-\delta_{R}}\right)^2\right)}}.
\label{eq:suff-cond-kth-step-temp}
\end{align}
We now derive sufficient conditions for success at
$k^{\mathrm{th}}$ step, $(k\ge 2)$, in the noiseless and noisy
measurement scenarios.
\begin{itemize}
\item For the noiseless case, putting $\bvec{e=0}$ in both sides
of the inequality in Eq~\eqref{eq:suff-cond-kth-step-temp}, we
obtain a sufficient condition for success in the noiseless case
as: \begin{align*} \frac{1}{\sqrt{K-c_k}}(1-\delta_{R}) &
\ge\frac{\delta_{R}}{\sqrt{L\left(1-\left(\frac{\delta_{R}}{1-\delta_{R}}\right)^2\right)}}.
\end{align*}
Using $\gamma:=\frac{\delta_{R}}{1-\delta_{R}}$, the above condition is seen to be satisfied if the following holds:\begin{align}
\sqrt{L(1-\gamma^2)} & >\gamma\sqrt{K-c_k}\nonumber\\
\Leftrightarrow\gamma & <\sqrt{\frac{L}{L+K-c_k}}\nonumber\\
\Leftrightarrow \delta_{LK+N-L+1} & <\frac{\sqrt{L}}{\sqrt{L}+\sqrt{L+K-c_k}}
\label{eq:recovery-condition-noiseless-k>2-temp}.
\end{align}
The above condition is ensured for all $k\ge 2$, if the following condition is satisfied, \begin{align}
\label{eq:recovery-condition-noiseless-k>2}
\delta_{LK+N-L+1} & <\frac{\sqrt{L}}{\sqrt{L}+\sqrt{L+K}}(<1/2).
\end{align}

\item For the noisy case, Eq.~\eqref{eq:suff-cond-kth-step-temp} is satisfied if the following is satisfied:
\begin{align}
 \label{eq:u-vs-v-inequlaity-recovery}
 \frac{\norm{\bvec{x}'_{T\cup T^k}}}{\norm{\bvec{e}}} 
 & \ge \frac{\sqrt{(1+\gamma)(1+2\gamma)}\left(\sqrt{K-c_k}+\sqrt{L(1-\gamma^2)}\right)}{\sqrt{L(1-\gamma^2)}-\gamma\sqrt{K-c_k}},
 \end{align}
 with the condition in Eq.~\eqref{eq:recovery-condition-noiseless-k>2-temp} assumed to hold.
The above lower bound can be simplified further by noting that
\begin{align*}
&RHS\;of\;(24)\;<&\\
\lefteqn{\sqrt{(1+\gamma)(1+2\gamma)}\frac{\sqrt{K}+\sqrt{L(1-\gamma^2)}}{\sqrt{L(1-\gamma^2)}-\gamma\sqrt{K}}} & &\\
\ & = \sqrt{\frac{1}{1-\delta_R}\cdot \frac{1+\delta_R}{1-\delta_R}}\cdot\frac{\frac{\sqrt{K}(1-\delta_R) + \sqrt{L(1-2\delta_R)}}{1-\delta_R}}{\frac{\sqrt{L(1-2\delta_R)} - \delta_R\sqrt{K}}{1-\delta_R}}\\
\ & =  \frac{\sqrt{1+\delta_{R}}}{1-\delta_R}\frac{\sqrt{K}(1-\delta_R) + \sqrt{L(1-2\delta_R)}}{\sqrt{L(1-2\delta_R)} - \delta_R\sqrt{K}}\\
\ & <\frac{ \sqrt{1+\delta_{R}}(\sqrt{K}+
\sqrt{L})}{\sqrt{L(1-2\delta_R)} - \delta_R\sqrt{K}} ,
\end{align*}
since $\sqrt{L(1-2\delta_R)}< \sqrt{L}(1-\delta_R)$.
 Thus, a modified condition for success at the $(k+1)^{th}$ iteration which also implies \eqref{eq:u-vs-v-inequlaity-recovery}
 is given by\begin{align}
 \label{eq:recovery-condition-noisy-k>2-temp}
 \frac{\norm{\bvec{x}'_{T\cup T^k}}}{\norm{\bvec{e}}} & >\frac{ \sqrt{1+\delta_{R}}(\sqrt{K}+ \sqrt{L})}{\sqrt{L(1-2\delta_R)} - \delta_R\sqrt{K}}.
 \end{align}
 Next, from the definition of $\kappa$ (section IV), \begin{align*}
 \norm{\bvec{x}'_{T\cup T^k}} & \ge \norm{\bvec{x}_{T\setminus T^k}}\ge |T\setminus T^k|\min_{j\in T}\abs{x_j}\\
 \ & = \norm{\bvec{x}}\cdot \kappa \cdot \sqrt{\frac{K-c_k}{K}}>\frac{\norm{\bvec{x}}\cdot \kappa}{\sqrt{K}},
 \end{align*}
 since $\min_{j\in T\setminus T^k}\abs{x_j}\ge \min_{j\in T}\abs{x_j}$ and $c_k<K$.
Combining with Eq.~\eqref{eq:recovery-condition-noisy-k>2-temp},
we obtain a sufficient condition for successful recovery at the
$k$-th step, $k\ge 2$ in the noisy measurement scenario as
    \begin{align}
 \label{eq:recovery-condition-noisy-k>2}
 \frac{\norm{\bvec{x}}}{\norm{\bvec{e}}}>\frac{\sqrt{1+\delta_{R}}(\sqrt{K}+ \sqrt{L})
 \sqrt{K}}{\kappa(\sqrt{L(1-2\delta_R)} - \delta_R\sqrt{K})},
 \end{align}
 \end{itemize}
 along with the condition in Eq~\eqref{eq:recovery-condition-noiseless-k>2}.
 \subsubsection{Condition for overall success}
 The condition for overall success is obtained by combining the conditions for
 success for $k=1$ and for $k\ge 2$, and is given below.\\
$\bullet$ For the noiseless scenario, a sufficient condition for
overall success has to comply with both the conditions in
Eq~\eqref{eq:recovery-condition-first-step-noiseless} and
Eq~\eqref{eq:recovery-condition-noiseless-k>2}. Since $R - (N+K) =
(L-1)(K-1)\ge 0$, as both $L,\ K$ are positive integers, we see
that the condition in
Eq~\eqref{eq:recovery-condition-noiseless-k>2} implies the
condition in
Eq~\eqref{eq:recovery-condition-first-step-noiseless}. Thus the
condition in Eq~\eqref{eq:recovery-condition-noiseless-k>2} serves
as a sufficient condition for overall success in noiseless
scenario. This proves Theorem~\ref{thm:noiseless-recovery}.\\
$\bullet$ For the noisy case, the conditions given by
\eqref{eq:recovery-condition-first-step-noisy} and
\eqref{eq:recovery-condition-noisy-k>2}, along with the conditions
given by \eqref{eq:recovery-condition-first-step-noiseless}, and
\eqref{eq:recovery-condition-noiseless-k>2} are sufficient. Of
these, we have already seen that
\eqref{eq:recovery-condition-noiseless-k>2} implies
\eqref{eq:recovery-condition-first-step-noiseless}. On the other
hand, it is easy to check that the numerator of the RHS of
~\eqref{eq:recovery-condition-noisy-k>2} is larger than that of
the RHS of \eqref{eq:recovery-condition-first-step-noisy}.
Further,
\begin{align*}
(1-2\delta_{LK+N-L+1})-(1-\delta_{N+K})^2 & =-\delta_{N+K}^2+2(\delta_{N+K}-\delta_{N+LK-L+1})<0,
\end{align*}
which implies that the denominator of the RHS of
\eqref{eq:recovery-condition-noisy-k>2} is smaller than that of
the RHS of ~\eqref{eq:recovery-condition-first-step-noisy}.
Moreover, by definition, $\kappa<1$. The overall implication of
these is that the condition in
\eqref{eq:recovery-condition-noisy-k>2} implies the condition in
~\eqref{eq:recovery-condition-first-step-noisy}. Finally, noting
that $\norm{\bvec{\Phi x}}\le
\sqrt{1+\delta_{K}}\norm{\bvec{x}}<\sqrt{1+\delta_{LK+N-L+1}}\norm{\bvec{x}}$,
the condition stated in Theorem~\eqref{thm:noisy-recovery}, along
with the condition in Theorem~\eqref{thm:noiseless-recovery} are
sufficient for overall successful recovery. This proves
Theorem~\ref{thm:noisy-recovery}.
\section{Proof of Lemma~\ref{lem:lemma-observations-average-energy}}
\label{sec:appendix-proof-lemma-observation-average-energy}
\begin{proof}

Let $N\le K$. Then, according to the definition of $S^1$(with $\bvec{r}^0$ given by $\bvec{y}$), we have for all $ \Lambda\subset T$ such that $|\Lambda|=N$, \begin{align*}
\norm{\bvec{\Phi}^t_{S^1}\bvec{y}}^2\ge & \norm{\bvec{\Phi}^t_{\Lambda}\bvec{y}}^2.
\end{align*}
Since there are $\binom{K}{N}$ such subsets of $T$, labelled, $\Lambda_i,\ 1\le i\le \binom{K}{N}$, we have \begin{align}
\label{eq:average-energy-inequality}
\binom{K}{N}\norm{\bvec{\Phi}^t_{S^1}\bvec{y}}^2\ge \sum_{i=1}^{\binom{K}{N}}\norm{\bvec{\Phi}^t_{\Lambda_i}\bvec{y}}^2.
\end{align}
Now, take any $j\in T$, and note that it appears in one of the $\Lambda_i$'s in exactly $\binom{K-1}{N-1}$ different ways. Thus, from the summation in Eq.~\eqref{eq:average-energy-inequality}, we find, \begin{align*}
\binom{K}{N}\norm{\bvec{\Phi}^t_{S^1}\bvec{y}}^2\ge \binom{K-1}{N-1}\norm{\bvec{\Phi}^t_{T}\bvec{y}}^2\\
\implies \norm{\bvec{\Phi}^t_{S^1}\bvec{y}}^2\ge \frac{N}{K}\norm{\bvec{\Phi}^t_{T}\bvec{y}}^2,
\end{align*}
from which Eq.~\eqref{eq:obs1} follows.

Now, let $N>K$. Then, we can take any subset $\Sigma\subset
\{1,2,\cdots,\ n\}$, such that $|\Sigma|=N$ and $T\subset \Sigma$.
Then, from definition, $\norm{\bvec{\Phi}^t_{S^1}\bvec{y}}^2\ge
\norm{\bvec{\Phi}^t_{\Sigma}\bvec{y}}^2\ge
\norm{\bvec{\Phi}^t_{T}\bvec{y}}^2$ from which Eq.~\eqref{eq:obs2}
follows.

To prove Eq.~\eqref{eq:obs3}, first note that
$\bvec{T^0}=\emptyset$ and thus, $\dualproj{T^0\cup
\{i\}}\bvec{y}=\dualproj{\{i\}}\bvec{y}=\bvec{y}-\frac{\inprod{\bvec{y}}
{\bvec{\phi}_i}}{\norm{\bvec{\phi}_i}^2}\bvec{\phi}_i=\bvec{y}-\inprod{\bvec{y}}{\bvec{\phi}_i}\bvec{y}$
(since $\norm{\bvec{\phi}}=1$), which means,
$\norm{\dualproj{T^0\cup
\{i\}}\bvec{y}}^2=\inprod{\dualproj{T^0\cup
\{i\}}\bvec{y}}{\bvec{y}}=\norm{\bvec{y}}^2-\abs{\inprod{\bvec{y}}{\bvec{\phi}_i}}^2$.
This means that $T^1$ consists of indices corresponding to the
largest $L$ absolute values
$\abs{\inprod{\bm{\phi}_i}{\bvec{y}}}^2$, for $i\in S^1$. But
since $S^1$ consists of indices corresponding to the $N$ largest
absolute values $\abs{\inprod{\bm{\phi}_i}{\bvec{y}}}^2$  with
$i\in \{1,2,\cdots,\ n\}=:\support$, and since $N\ge L$, we have,
$\min_{i\in T^1}\abs{\inprod{\bm{\phi}_i}{\bvec{y}}}^2\ge
\max_{i\in \support\setminus
T^1}\abs{\inprod{\bm{\phi}_i}{\bvec{y}}}^2$. Since, $L\le K$, for
each $\Gamma\subset T$, such that $|\Gamma|=L$, we have
\begin{align*} \norm{\bvec{\Phi}^t_{T^1}\bvec{y}}^2\ge
\norm{\bvec{\Phi}^t_\Gamma\bvec{y}}^2
\end{align*}
Since there are $\binom{K}{L}$ such subsets, we can write \begin{align*}
\binom{K}{L}\norm{\bvec{\Phi}^t_{T^1}\bvec{y}}^2\ge \sum_{\Gamma:\Gamma\subset T,\ |\Gamma|=L}\norm{\bvec{\Phi}^t_\Gamma\bvec{y}}^2
\end{align*}
Now any index $i\in T$ is contained in exactly $\binom{K-1}{L-1}$ of such $L$ cardinality subsets. Hence \begin{align*}
\binom{K}{L}\norm{\bvec{\Phi}^t_{T^1}\bvec{y}}^2\ge\binom{K-1}{L-1}\norm{\bvec{\Phi}^t_T\bvec{y}}^2,
\end{align*}
from which Eq.~\eqref{eq:obs3} follows.
\end{proof}
\bibliography{joint-omp-ols}

\begin{thebibliography}{10}
\providecommand{\url}[1]{#1}
\csname url@samestyle\endcsname
\providecommand{\newblock}{\relax}
\providecommand{\bibinfo}[2]{#2}
\providecommand{\BIBentrySTDinterwordspacing}{\spaceskip=0pt\relax}
\providecommand{\BIBentryALTinterwordstretchfactor}{4}
\providecommand{\BIBentryALTinterwordspacing}{\spaceskip=\fontdimen2\font plus
\BIBentryALTinterwordstretchfactor\fontdimen3\font minus
  \fontdimen4\font\relax}
\providecommand{\BIBforeignlanguage}[2]{{%
\expandafter\ifx\csname l@#1\endcsname\relax
\typeout{** WARNING: IEEEtran.bst: No hyphenation pattern has been}%
\typeout{** loaded for the language `#1'. Using the pattern for}%
\typeout{** the default language instead.}%
\else
\language=\csname l@#1\endcsname
\fi
#2}}
\providecommand{\BIBdecl}{\relax}
\BIBdecl

\bibitem{tropp2007signal}
J.~A. Tropp and A.~C. Gilbert, ``Signal {R}ecovery {F}rom {R}andom
  {M}easurements {V}ia {O}rthogonal {M}atching {P}ursuit,'' \emph{IEEE Trans.
  Inf. Theory}, vol.~53, no.~12, pp. 4655--4666, 2007.

\bibitem{wang2012generalized}
J.~Wang, S.~Kwon, and B.~Shim, ``Generalized orthogonal matching pursuit,''
  \emph{IEEE Trans. Signal Process.}, vol.~60, no.~12, pp. 6202--6216, 2012.

\bibitem{soussen2013joint}
C.~Soussen, R.~Gribonval, J.~Idier, and C.~Herzet, ``Joint k-step {A}nalysis of
  {O}rthogonal {M}atching {P}ursuit and {O}rthogonal {L}east {S}quares,''
  \emph{IEEE Trans. Inf. Theory}, vol.~59, no.~5, pp. 3158--3174, 2013.

\bibitem{chen1989orthogonal}
S.~Chen, S.~A. Billings, and W.~Luo, ``Orthogonal {L}east {S}quares {M}ethods
  and {T}heir {A}pplication to {N}on-linear {S}ystem {I}dentification,''
  \emph{Int. J. Control}, vol.~50, no.~5, pp. 1873--1896, 1989.

\bibitem{wang2017recovery}
J.~Wang and P.~Li, ``Recovery of {S}parse {S}ignals {U}sing {M}ultiple
  {O}rthogonal {L}east {S}quares,'' \emph{IEEE Trans. Signal Process.},
  vol.~65, no.~8, pp. 2049--2062, April 2017.

\bibitem{candes2006robust}
E.~J. Cand{\`e}s, J.~Romberg, and T.~Tao, ``Robust uncertainty principles:
  Exact signal reconstruction from highly incomplete frequency information,''
  \emph{IEEE Trans. Inf. Theory}, vol.~52, no.~2, pp. 489--509, 2006.

\bibitem{dai2009subspace}
W.~Dai and O.~Milenkovic, ``Subspace pursuit for compressive sensing signal
  reconstruction,'' \emph{IEEE Trans. Inf. Theory}, vol.~55, no.~5, pp.
  2230--2249, 2009.

\bibitem{needell2009cosamp}
D.~Needell and J.~A. Tropp, ``Cosamp: Iterative signal recovery from incomplete
  and inaccurate samples,'' \emph{Appl. Comput. Harmon. Anal.}, vol.~26, no.~3,
  pp. 301--321, 2009.

\bibitem{satpathi2013improving}
S.~Satpathi, R.~L. Das, and M.~Chakraborty, ``Improving the bound on the rip
  constant in generalized orthogonal matching pursuit,'' \emph{IEEE Signal
  Process. Lett.}, vol.~20, no.~11, pp. 1074--1077, 2013.

\bibitem{rebollo2002optimized}
L.~Rebollo-Neira and D.~Lowe, ``Optimized orthogonal matching pursuit
  approach,'' \emph{IEEE Signal Process. Lett.}, vol.~9, no.~4, pp. 137--140,
  2002.

\bibitem{fletcher2012orthogonal}
A.~K. Fletcher and S.~Rangan, ``Orthogonal matching pursuit: A brownian motion
  analysis,'' \emph{IEEE Trans. Signal Process.}, vol.~60, no.~3, pp.
  1010--1021, 2012.

\bibitem{wen2017novel}
J.~Wen, Z.~Zhou, D.~Li, and X.~Tang, ``A novel sufficient condition for
  generalized orthogonal matching pursuit,'' \emph{IEEE Commun. Lett.},
  vol.~21, no.~4, pp. 805--808, 2017.

\bibitem{wen2017nearly}
J.~Wen, J.~Wang, and Q.~Zhang, ``Nearly optimal bounds for orthogonal least
  squares,'' \emph{IEEE Trans. Signal Process.}, vol.~65, no.~20, pp.
  5347--5356, 2017.

\bibitem{li2015sufficient}
B.~Li, Y.~Shen, Z.~Wu, and J.~Li, ``Sufficient conditions for generalized
  orthogonal matching pursuit in noisy case,'' \emph{Signal Process.}, vol.
  108, pp. 111--123, 2015.

\bibitem{golub2012matrix}
G.~H. Golub and C.~F. Van~Loan, \emph{Matrix {C}omputations}.\hskip 1em plus
  0.5em minus 0.4em\relax JHU Press, 2012, vol.~3.

\end{thebibliography}
\end{document}